\documentclass[10pt, oneside]{article}
\usepackage[english]{babel}
\usepackage{amsmath}
\usepackage{amsthm}
\usepackage{amssymb}
\usepackage{graphicx}
\usepackage{pst-all}
\usepackage{authblk}
\usepackage{enumerate}
\usepackage{graphics}
\usepackage{xspace,colortbl}
\usepackage[latin1]{inputenc}
\usepackage{float}

\usepackage{fancyvrb}
\DefineVerbatimEnvironment
{MiVerbatim}{Verbatim}{fontsize=\footnotesize,frame=single,label=\emph{Mathematica
8.0},framesep=2mm,numbers=left}

\title{A Generalization of the Fibonacci Word Fractal and the Fibonacci Snowflake}
\author[1]{José L. Ramírez \footnote{Corresponding author.}  \thanks{josel.ramirez@ima.usergioarboleda.edu.co}}
\author[2]{Gustavo N. Rubiano \thanks{gnrubianoo@unal.edu.co} }
\author[2]{Rodrigo de Castro \thanks{rdecastrok@unal.edu.co}}
\affil[1]{\small Instituto de Matemáticas y sus Aplicaciones,  Universidad Sergio Arboleda, Calle 74 no. 14 - 14, Bogotá, Colombia }
\affil[2]{\small Departamento de Matemáticas, Universidad Nacional de Colombia, AA 14490, Bogotá, Colombia}

\begin{document}
\newtheorem{theorem}{Theorem}
\newtheorem{definition}{Definition}
\newtheorem{corollary}{Corollary}
\newtheorem{example}{Example}
\newtheorem{lemma}{Lemma}
\newtheorem{proposition}{Proposition}
\maketitle
\setlength{\parindent}{0pt}

\begin{abstract}

In this paper we introduce a family of infinite words that generalize the Fibonacci word and we study their combinatorial properties. We associate with this family of words a family of curves that are like the Fibonacci word fractal and reveal some fractal features.  Finally, we describe an infinite family of polyominoes stems from the generalized Fibonacci words and we study some of their geometric properties, such as perimeter and area.  These last polyominoes generalize the Fibonacci snowflake and they are double squares polyominoes, i.e., tile the plane by translation in exactly two distinct ways.   \\


\end{abstract}
\date
\maketitle

\section{Introduction}
The infinite Fibonacci word,
\begin{align*}
\textbf{\emph{f}}=\texttt{0100101001001010010100100101}\cdots
\end{align*}

is certainly one of the most studied examples in the combinatorial theory of infinite words, e.g.  \cite{BER, FIB1, FIB2, FIB3, LUCA, DRO, FIB5, PIR}. It is the archetype of a Sturmian word  \cite{LOT2}. The Fibonacci word \emph{\textbf{f}} can be defined in several different ways \cite{BER}. For instance, Fibonacci word \textbf{\textit{f}} satisfies  $\lim_{n\rightarrow\infty}\sigma^n(\verb"1")=\textbf{\textit{f}}$, where  $\sigma:\left\{\texttt{0,1}\right\}\rightarrow \left\{\texttt{0,1}\right\}$ is the morphism defined by  $\sigma(\verb"0")=\verb"01"$ and $\sigma(\verb"1")=\verb"0"$. This morphism is called \emph{Fibonacci morphism}. The name Fibonacci given to \textbf{\textit{f}} is due to the fact that \textbf{\textit{f}} is the limit sequence of the infinite sequence $(f_n)_{n=0}$ of finite words over $\left\{\texttt{0,1}\right\}$ defined inductively as follows
\begin{align*}
f_0=\texttt{1},  \hspace{1cm} f_1=\texttt{0}, \hspace{1cm} f_n=f_{n-1}f_{n-2}, \ n\geq 2.
\end{align*}
The words $f_n$ are called \emph{finite Fibonacci words}.  It is clear that $|f_n|=F_n$, where $F_n$ is the $n$-th Fibonacci number  defined by the recurrence relation $F_n=F_{n-1}+F_{n-2}$, for all integer $n\geq 2$ and with initial values $F_0=1=F_1$.

The word \textbf{\textit{f}} can be associated with a curve from a drawing rule, which has geometry properties obtained from combinatorial properties of \textbf{\emph{f}} \cite{BLO2, ALE}. We must travel the word in a particular way, depending on the symbol read a particular action is produced, this idea is the same as that used in the L-Systems  \cite{PUB}. In this case, the drawing rule is called ``odd-even drawing rule''  \cite{ALE},  this is defined as shown in the following table:

\begin{center}
\begin{tabular}{|p{1.8cm}|p{10cm}|}\hline
 \centering \textbf{Symbol}  &  \textbf{Action} \\ \hline
\centering
 \verb"1"
 & Draw a line forward. \\ \hline
\centering
 \verb"0"
& Draw a line forward and if the symbol \verb"0" is in an even position  then turn left and if \texttt{0} is in an odd position  then turn right.\\ \hline
\end{tabular}
\end{center}

The \emph{$n$th-curve  of Fibonacci}, denoted by $\mathcal{F}_n$, is obtained by applying the odd-even drawing rule to the word $f_n$. The \emph{Fibonacci word fractal} $\mathcal{F}$, is defined as
\begin{align*}
\mathcal{F}=\lim_{n\rightarrow \infty}\mathcal{F}_n.
\end{align*}

For example, in Fig. \ref{graf22} we show the curve $\mathcal{F}_{10}$ and $\mathcal{F}_{17}$. The graphics in this paper were generated using the software \verb"Mathematica 8.0", \cite{RAM}.
\begin{center}
$f_{10}=$\verb"010010100100101001010010010100100101001010010010"\\ \verb"10010100100101001001010010100100101001001".
\end{center}
\begin{figure}[H]
\centering
  \includegraphics[scale=0.5]{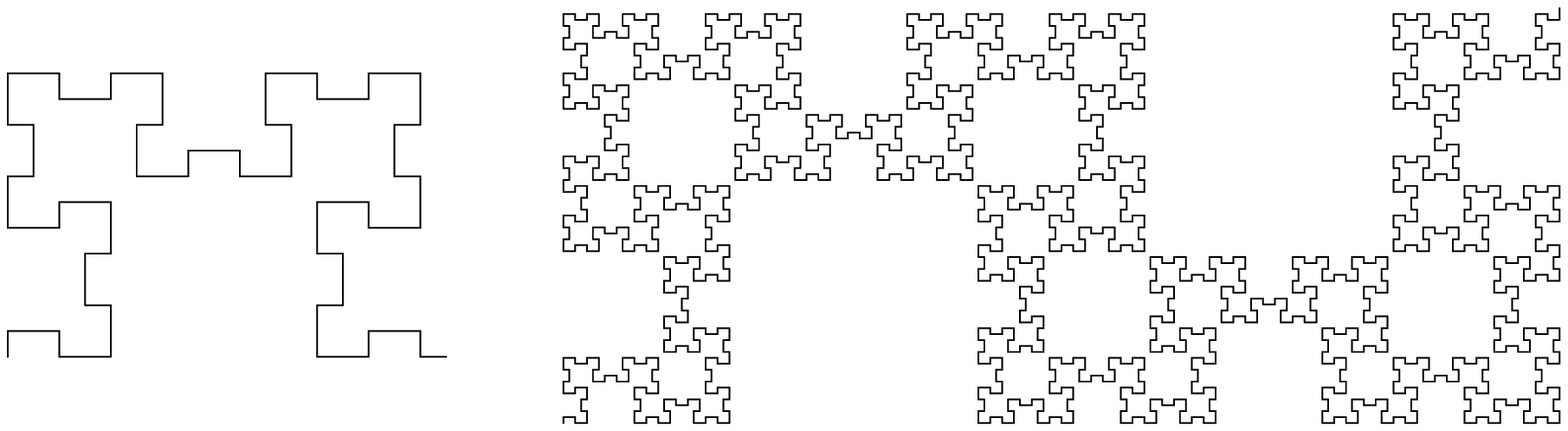}
  \caption{Fibonacci curves $\mathcal{F}_{10}$ and $\mathcal{F}_{17}$ corresponding to the words $f_{10}$ and $f_{17}$.} \label{graf22}
\end{figure}

The word \textbf{\emph{f}} can also be associated with a family of polyominoes  which tile the plane by translation  and are called  \emph{Fibonacci snowflakes} \cite{BLO2, FIB}. By \emph{polyomino} we mean  a finite union of unit lattice squares in the  square lattice $\mathbb{Z}\times\mathbb{Z}$ whose boundary is a non-crossing closed path (see \cite{GRU} for more on tilings and \cite{BRA} for related problems).  A \emph{path }in the square lattice is a polygonal path made of the elementary unit translations
\begin{align*}
\texttt{0}=(1,0), \hspace{1cm} \texttt{1}=(0,1), \hspace{1cm} \texttt{2}=(-1,0), \hspace{1cm} \texttt{3}=(0,-1).
\end{align*}
These paths are conveniently encoded by words on the alphabet  $\mathcal{A}=\left\{\verb"0, 1, 2, 3"\right\}$.  This relation between discrete objects and words  has been used in modeling of problems  of tessellations in the plane with polyominoes, (see e.g. \cite{NIV, BLO2, BLO, BLO3, BRL} and \cite{BRL2}  for more relations between discrete geometry and combinatorics on words).

In \cite{WIJ}  authors were the first to consider the problem of deciding if a given polyomino tiles the plane by translation and they coined the term \emph{exact polyomino}.   In \cite{NIV}   authors proved that a polyomino $P$ tiles the plane by translations if and only if the boundary word \textbf{b}($P$) is equal up to a cyclic permutation of the symbols to $A\cdot B \cdot C \cdot \widehat{A}\cdot \widehat{B} \cdot \widehat{C}$, where one of the variables in the factorization may be empty. This condition is referred as the BN-factorization. If the boundary word is equal to  $A\cdot B \cdot C \cdot \widehat{A}\cdot \widehat{B} \cdot \widehat{C}$  such a polyomino is called \emph{pseudo-hexagon} and when one of the variables is empty, i.e., $\textbf{b}(P)=A\cdot B \cdot \widehat{A}\cdot \widehat{B}$, we say that $P$ is a \emph{square polyomino}.

For instance, the polyomino in Fig. \ref{pol2} (left) is an exact polyomino and its boundary can be  factorized by  $\verb"122"\cdot\verb"212"\cdot\verb"323"\cdot\verb"003"\cdot\verb"030"\cdot\verb"101"$, (the factorization is not necessarily in a unique way).
\begin{figure}[!h]
\centering
 \includegraphics[scale=0.4]{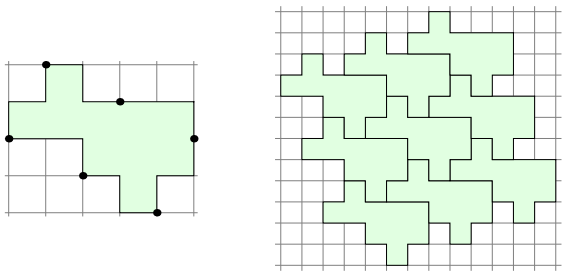}\\
  \caption{Exact polyomino and tiling.}\label{pol2}
\end{figure}

In \cite{BLO}, authors  prove that an exact polyomino tiles the plane in at most two distinct ways. Squares polyominoes having exactly two distinct BN-factorizations are called \emph{double squares}. For  instance, Christoffel and Fibonacci tiles or Fibonacci snowflakes, introduced in \cite{BLO2}, are examples of double squares, however, there exist double squares not in the Christoffel and Fibonacci tiles families. In \cite{BLO3}, they study the combinatorial properties and the problem of generating exhaustively double square tiles, however, they did not study the geometric properties, only in the case of Fibonacci polyominoes \cite{FIB}. \\

On the other hand, Fibonacci numbers and their generalizations have many
interesting properties and applications to almost every field of science and art,  (e.g. see \cite{koshy}). In the present case we are interested in the generalization of the Fibonacci sequence by preserving the recurrence relation and altering the first two terms of the sequence.\\

The \emph{$(n,i)$-th Fibonacci number} $F_n^{\left[i\right]}$ is defined recursively by 
\begin{align*}
F_0^{\left[i\right]}=1, \hspace{1cm} F_1^{\left[i\right]}=i,   \hspace{1cm} F_n^{\left[i\right]}=F_{n-1}^{\left[i\right]}+F_{n-2}^{\left[i\right]} 
\end{align*} 
for all $n\geq2$ and $i\geq1$.  For $i=1, 2$ we have the Fibonacci numbers. \\

In this paper we introduce a family of words $\textbf{\emph{f}}^{\left[i\right]}$ (Definition \ref{genfibo}) that generalize the Fibonacci word.  Each word $\textbf{\emph{f}}^{\left[i\right]}$ is the limit sequence of an infinite sequence of finite words such that their length are  $(n,i)$-th Fibonacci numbers. Moreover, the word   $\textbf{\emph{f}}^{\left[i\right]}$  is  a  characteristic  word of slope $\frac{i-\phi}{i^2-i-1}$, where $\phi$ is the golden ratio (Theorem \ref{slopefibo}). From this family of infinite words we define a family of plane curves called\emph{ $i$-Fibonacci word fractal }(Definition \ref{ifractalfibo}), which are like the Fibonacci word fractal and have the same properties (Proposition \ref{ifibofractal}). Finally, we  introduce a family of  polyominoes which generalize the Fibonacci snowflake and we study their geometric properties, such as perimeter  (Proposition \ref{perimetro}) and area (Proposition \ref{area}) which is related to generalized Pell numbers. These polyominoes are also double squares (Theorem \ref{poldobles}) and have the same fractal dimension of the Fibonacci word Fractal.  These generalizations are interesting, as they leave the question whether it is possible to generate all double squares polyominoes from families of  words like the Fibonacci word.

\section{Definitions and Notation}

The terminology and notations are mainly those of Lothaire \cite{LOT2} and Allouche and Shallit \cite{SHA2}. Let $\Sigma$ be a finite alphabet, whose elements are called \emph{symbols}. A \emph{word} over  $\Sigma$  is a finite sequence of symbols from $\Sigma$. The set of all words over $\Sigma$, i.e., the free monoid generated by $\Sigma$, is denoted by $\Sigma^*$. The identity element $\epsilon$ of $\Sigma^*$ is called the \emph{empty word}. For any word $w\in\Sigma^*$,  $\left|w\right|$  denotes its \emph{length}, i.e., the number of symbols occurring in $w$. The length of $\epsilon$ is taken to be equal to 0. If $a\in\Sigma$ and $w\in \Sigma^{*}$, then $\left|w\right|_{a}$ denotes the number of occurrences of $a$ in $w$.\\

For two words $u=a_{1}a_{2}\cdots a_{k}$ and $v=b_{1}b_{2}\cdots b_{s}$ in $\Sigma^{*}$ we denote by $uv$ the \emph{concatenation} of the two words, that is,  $uv=a_{1}a_{2}\cdots a_{k}b_{1}b_{2}\cdots b_{s}$. If $v=\epsilon$ then $u\epsilon=\epsilon u=u$, moreover, by $u^n$ we denote the word $uu \cdots u$ ($n$ times). A word  $v$ is a \emph{factor} or \emph{subword} of $u$ if there exist  $x, y \in \Sigma^{*}$  such that $u = xvy$. If $x=\epsilon$ ($y=\epsilon$), then $v$ is called \emph{prefix} (\emph{suffix}) of $u$.\\

The \emph{reversal} of a word $u=a_1a_2\cdots a_n$ is the word  $u^R=a_{n}\cdots a_2a_{1}$ and $\epsilon^R=\epsilon$. A word $u$ is a \emph{palindrome} if $u^R=u$.\\

An \emph{infinite word} over $\Sigma$ is a map $\textbf{\emph{u}}:\mathbb{N}\rightarrow\Sigma$. It is written $\textbf{\emph{u}}=a_1a_2a_3\ldots$. The set of all infinite words over  $\Sigma$ is denoted by $\Sigma^{\omega}$.

\begin{example}\label{eje1}
 Let $\textbf{p}=(p_n)_{n\geq1}=\verb"0110101000101"\cdots$,  where $p_n=\texttt{\emph{1}}$ if $n$ is a prime number and $p_n=\texttt{\emph{0}}$ otherwise, is an example of an infinite word. The word $\textbf{p}$ is called
 the characteristic sequence of the prime numbers.
\end{example}

Let $\Sigma$ and $\Delta$ be alphabets. A \emph{morphism} is a map $h:\Sigma^*\rightarrow \Delta^* $ such that $h(xy)=h(x)h(y)$ for all $x, y\in\Sigma^*$. It is clear that $h(\epsilon)=\epsilon$. Furthermore, a morphism is completely determined by its action on  single symbols.

There is a special class of infinite words, with many remarkable properties, the so-called Sturmian words. These words admit several equivalent definitions (see, e.g. \cite{SHA2} or \cite{LOT2}). Let  $\textbf{\textit{w}}\in\Sigma^{\omega}$. We define $P(\textbf{\textit{w}},n)$, the \emph{complexity function} of $\textbf{\textit{w}}$, to be the map  that counts, for all integer $n\geq 0$, the number of subwords of length $n$ in  $\textbf{\textit{w}}$. An infinite word $\textbf{\textit{w}}$ is a \emph{Sturmian word} if  $P(\textbf{\textit{w}},n)=n+1$ for all integer $n\geq 0$.
Since for any Sturmian word  $P(\textbf{\textit{w}},1)=2$, then  Sturmian words are over two symbols. The word $\textbf{\emph{p}}$, in example \ref{eje1}, is not a Sturmian word because  $P(\textbf{\emph{p}}, 2)=4$.\\

Given two real numbers $\alpha, \beta\in \Bbb R$ with $\alpha$ irrational and $0<\alpha < 1$,  $0\leq \beta < 1$, we define the infinite word $\textbf{\emph{w}}=w_1w_2 w_3\cdots $ as  $$w_n= \lfloor (n+1) \alpha  +\beta \rfloor - \lfloor n\alpha +\beta \rfloor.$$
The numbers $\alpha$ and $\beta$ are called the \emph{slope} and the \emph{intercept}, respectively. Words of this form are called \emph{lower mechanical words} and are known to be equivalent to Sturmian words \cite{LOT2}. As special case,  when  $\beta=0$, we obtain the \emph{characteristic words}.

\begin{definition}
Let $\alpha$ be an irrational number with $0<\alpha<1$. For $n\geq1$, define
\begin{align*}
w_\alpha(n):=\left\lfloor(n+1)\alpha\right\rfloor - \left\lfloor n\alpha\right\rfloor
\end{align*}
and
\begin{align*}
\textbf{w}(\alpha):=w_\alpha(1)w_\alpha(2)w_\alpha(3)\cdots
\end{align*}
Then $\textbf{w}(\alpha)$ is called the characteristic word with slope $\alpha$.
 \end{definition}

On the other hand, note that every irrational $\alpha \in (0, 1)$ has a unique continued fraction expansion
\begin{align*}
\alpha=\left[0, a_1, a_2, a_3, \ldots\right]=\cfrac{1}{a_1 + \cfrac{1}{a_2 + \cfrac{1}{a_3 + \cdots}}}
\end{align*}
where each $a_i$ is a positive integer. Let $\alpha=\left[0, 1 + d_{1}, d_{2}, \dots \right]$ be an irrational number with $d_{1}\geq0$ and $d_{n}>0$ for $n>1$. With the directive sequence $(d_{1}, d_{2}, \dots, d_{n}, \dots)$, we associate a sequence $(s_{n})_{n\geq-1 }$ of words defined by
\begin{align*}
s_{-1}=\texttt{1}, \ \ s_{0}=\texttt{0}, \ \  s_{n}=s_{n-1}^{d_{n}}s_{n-2}, \ \ (n\geq 1)
\end{align*}
Such a sequence of words is called a \emph{standard sequence}. This sequence is related to characteristic words in the following way. Observe that, for any $n\geq 0$, $s_n$ is a prefix of $s_{n+1}$, which gives meaning to $\lim_{n\rightarrow\infty}s_{n}$ as an infinite word. In fact, one can prove \cite{LOT2} that each  $s_{n}$ is a prefix of $\textbf{\emph{w}}(\alpha)$  for all $n\geq 0$ and
\begin{align}\label{fraccon}
\textbf{\emph{w}}(\alpha)=\lim_{n\rightarrow\infty}s_{n}.
\end{align}

\subsection{Fibonacci Word and Its Fractal Curve}

The infinite Fibonacci word $\textbf{\emph{f}}$ is a Sturmian word \cite{LOT2}, more precisely, $\textbf{\emph{f}}=\textbf{\textit{w}}\left(\frac{1}{\phi^2}\right)$ where $\phi= \frac{1+\sqrt{5}}{2}$ is the golden ratio.\\

Let $\Phi:\left\{\verb"0", \verb"1"\right\}^{*}\rightarrow \left\{\verb"0", \verb"1"\right\}^{*}$ be a map such that $\Phi$ deletes the last two symbols, i.e., $\Phi(a_1a_2\cdots a_{n})=a_1a_2\cdots a_{n-2}$, $(n\geq 2)$.\\

The following proposition summarizes some basic properties about Fibonacci word.
\begin{proposition}[Pirillo \cite{PIR}]\label{fibo}
The Fibonacci word  and the finite Fibonacci words, satisfy the following properties
\begin{enumerate}[i.]
\item The words \verb"11" and \verb"000" are not subwords of the Fibonacci word.

  \item For all $n\geq 2$. Let $ab$ be the last two symbols of $f_n$, then we have  $ab = \verb"01"$ if  $n$ is even and $ab = \verb"10"$ if $n$ is odd.
  \item The concatenation of two successive Fibonacci words is ``almost commutative'', i.e., $f_{n}f_{n-1}$ and $f_{n-1}f_{n}$ have a common prefix of length  $F_{n}-2$ for all $n\geq2$.
  \item $\Phi(f_n)$ is a palindrome for all  $n\geq2$.
   \item For all $n\geq6$, $f_n = f_{n-3}f_{n-3}f_{n-6}l_{n-3}l_{n-3}$, where $l_n=\Phi(f_n)ba$, i.e.,  $l_n$ exchanges the two last symbols of $f_n$.
\end{enumerate}
\end{proposition}

In the next proposition we show some properties of the curves $\mathcal{F}_{n}$ and $\mathcal{F}$. It comes directly from the properties of the Fibonacci word, see Proposition \ref{fibo}.

\begin{proposition}[Monnerot \cite{ALE}]\label{fractalfibo}
Fibonacci word fractal $\mathcal{F}$ and the curve $\mathcal{F}_n$  have the following properties:
\begin{enumerate}[i.]
  \item $\mathcal{F}$ is composed only of segments of length 1
or 2.
  \item The  curve $\mathcal{F}_n$ is similar to the curve $\mathcal{F}_{n-3}$, i.e., they have the same shape except for the number of segments.
 \item The curve $\mathcal{F}_{n}$ is symmetric. More precisely, the curves $\mathcal{F}_{3n}$ and $\mathcal{F}_{3n+1}$ are symmetric with respect to a line and $\mathcal{F}_{3n+2}$ is symmetric with respect to a point.
  \item The curve $\mathcal{F}_n$ is composed of 5 curves: $\mathcal{F}_{n}=\mathcal{F}_{n-3}\mathcal{F}_{n-3}\mathcal{F}_{n-6}\mathcal{F'}_{n-3}\mathcal{F'}_{n-3}$, where $\mathcal{F'}_{n}$ is obtained by applying the odd-even drawing rule to word $l_n$, see Proposition  \ref{fibo}-v.
\item The fractal  dimension of the Fibonacci word fractal is
\begin{align*}
3\frac{\log \phi}{\log (1+\sqrt 2)}=1.6379\dots
\end{align*}
\end{enumerate}
\end{proposition}
More of these properties can be found in \cite{ALE}.

\section{Generalized Fibonacci Words and Fibonacci Word Fractals}

In this section, we introduce a generalization of the Fibonacci word and the Fibonacci word fractal, and we show that Propositions \ref{fibo} and \ref{fractalfibo} remain.
\begin{definition}\label{genfibo}
The $(n,i)$-Fibonacci words are words over $\left\{\texttt{\emph{0,1}}\right\}$ defined inductively as follows
\begin{align*}
f_0^{\left[i\right]}=\texttt{\emph{0}}, \hspace{1cm} f_1^{\left[i\right]}=\texttt{\emph{0}}^{i-1}\texttt{\emph{1}},  \hspace{1cm} f_n^{\left[i\right]}=f_{n-1}^{\left[i\right]}f_{n-2}^{\left[i\right]},
\end{align*}
for all $n\geq 2$ and $i\geq1$. The infinite word
 \begin{align*}
 \textbf{f}^{\, \left[i\right]}:=\lim_{n\rightarrow\infty}f_n^{\left[i\right]}
  \end{align*}
   is called  the $i$-Fibonacci word.
\end{definition}

For $i=2$ we have the classical  Fibonacci word.

\begin{example}  The first $i$-Fibonacci words are
\begin{center}
\begin{tabular}{lll}
  $\textbf{f}^{\,\left[1\right]}=\verb"1011010110110"\cdots=\overline{\textbf{f}\ }$, & $\textbf{f}^{\,\left[2\right]}=\verb"0100101001001"\cdots=\textbf{f}$, &
  $\textbf{f}^{\,\left[3\right]}=\verb"0010001001000"\cdots$, \\
  $\textbf{f}^{\,\left[4\right]}=\verb"0001000010001"\cdots$, &
  $\textbf{f}^{\,\left[5\right]}=\verb"0000100000100"\cdots$, &
    $\textbf{f}^{\,\left[6\right]}=\verb"0000010000001"\cdots$ \\
\end{tabular}
\end{center}
\end{example}

Note that the length of the word $f_n^{\left[i\right]}$ is the  $(n,i)-$th Fibonacci number $F_n^{\left[i\right]}$, i.e., $|f_n^{\left[i\right]}|=F_n^{\left[i\right]}$. It is clear because  $f_n^{\left[i\right]}=f_{n-1}^{\left[i\right]}f_{n-2}^{\left[i\right]}$ and then  $|f_n^{\left[i\right]}|=|f_{n-1}^{\left[i\right]}| + |f_{n-2}^{\left[i\right]}|$, moreover  $|f_{0}^{\left[i\right]}|=1$ and $|f_{1}^{\left[i\right]}|=i$.

\begin{proposition}
A formula for the  $(n, i)$-th Fibonacci number is
\begin{align*}
F_n^{\left[i\right]}=\frac{1}{2\sqrt{5}}\left(\left(\frac{1-\sqrt{5}}{2}\right)^n(\sqrt{5}+1-2i) + \left(\frac{1+\sqrt{5}}{2}\right)^n(\sqrt{5}-1+2i) \right).
\end{align*}
\end{proposition}
\begin{proof}
The proof is by induction on $n$. This is clearly true for $n=0, 1$. Now suppose the result is true for $n$. Then
\begin{align*}
F_{n+1}^{\left[i\right]}&=F_{n}^{\left[i\right]} + F_{n-1}^{\left[i\right]} =\frac{1}{2\sqrt{5}}\left(\left(\phi_{1}^{n}+\phi_{1}^{n-1}\right)\left(\sqrt{5}+1-2i \right)  + \left(\phi_{2}^{n}+\phi_{2}^{n-1} \right) \left(\sqrt{5}-1+2i \right)\right)
\end{align*}
where $\phi_{1}=\frac{1-\sqrt{5}}{2}$ and $\phi_{2}=\frac{1+\sqrt{5}}{2}$. Moreover,
\begin{align*}
\phi_{1}^{n}+\phi_{1}^{n-1} =  \phi_{1}^{n-1}(\phi_{1}+1)=\phi_{1}^{n-1}\left(\frac{1-\sqrt{5}}{2}+1\right)=\phi_{1}^{n-1}\phi_{1}^{2}=\phi_{1}^{n+1},
\end{align*}
analogously $\phi_{2}^{n}+\phi_{2}^{n-1}=\phi_{2}^{n+1}$.
So \begin{align*}
F_{n+1}^{\left[i\right]}&=\frac{1}{2\sqrt{5}}\left(\phi_{1}^{n+1}\left(\sqrt{5}+1-2i \right)  + \phi_{2}^{n+1} \left(\sqrt{5}-1+2i \right)\right). \qedhere
\end{align*}
\end{proof}

Table \ref{numfibo} shows the first numbers  $F_n^{\left[i\right]}$ and their coincidence with some remarkable sequences in the OIES \footnote{
Many integer sequences and their properties are to be found electronically on the On-Line Encyclopedia of Sequences, \cite{OEIS}.}.
\begin{table}[h]
\centering
\begin{tabular}{|c|lc|}  \hline
  $i$ & \multicolumn{2}{c|}{ $\left\{ F_n^{\left[i\right]}\right\}_{n\geq0}$ }\\ \hline
  1 & $\left\{1, 1, 2, 3, 5, 8, 13, 21, 34, 55, 89, 144,...\right\}$,  & (A000045). \\ \hline
  2 & $\left\{1, 2, 3, 5, 8, 13, 21, 34, 55, 89, 144, 233,...\right\}$,  &(A000045). \\ \hline
  3 & $\left\{1,3,4,7,11,18,29,47,76,123,199,322,...\right\}$, & (A000204). \\ \hline
  4 & $\left\{1,4,5,9,14,23,37,60,97,157,254,411,...\right\}$, & (A000285). \\ \hline
  5 & $\left\{1,5,6,11,17,28,45,73,118,191,309,500,...\right\}$,&  (A022095). \\ \hline
  6 &  $\left\{1,6,7,13,20,33,53,86,139,225,364,589,...\right\}$, & (A022096). \\  \hline
\end{tabular}
\caption{First numbers $F_n^{\left[i\right]}$.}\label{numfibo}
\end{table}

The following proposition relates the Fibonacci word  $\textbf{\emph{f}}$ with $\textbf{\emph{f}}^{\left[i\right]}$.
\begin{proposition}\label{morfismo}
Let $\varphi_i:\left\{\verb"0", \verb"1"\right\}^{*}\rightarrow \left\{\verb"0", \verb"1"\right\}^{*}$ be the morphism defined by  $\varphi_i(\verb"0")=\verb"0"$ and $\varphi_i(\verb"1")=\verb"0"^i\verb"1"$,  $i\geq 0$, then
\begin{align*}
\textbf{f}^{\,\left[i+2\right]}=\varphi_{i}\left(\textbf{f}\right)
\end{align*}
for all $i\geq 0$.
\end{proposition}

\begin{proof}
It suffices to prove that $\emph{f}_{n-1}^{\,\left[i+2\right]}=\varphi_i(f_n)$ for all integers $n\geq2$ and $i\geq 0$.  We prove this by induction on $n$. For $n=2$ we have $\varphi_i(f_2)=\varphi_i(\texttt{01})=\texttt{0}^{i+1}\texttt{1}=\emph{f}_1^{\,\left[i+2\right]}$. Now suppose the result is true for $n$.  Then $\varphi_i(f_{n+1})=\varphi_i(f_{n}f_{n-1})=\varphi_i(\emph{f}_{n})\varphi_i(\emph{f}_{n-1})=\emph{f}_{n-1}^{\,\left[i+2\right]}\emph{f}_{n-2}^{\,\left[i+2\right]}=\emph{f}_{n}^{\,\left[i+2\right]} $. \end{proof}

The following proposition generalizes Proposition \ref{fibo}.

\begin{proposition}\label{fibog}
The $i$-Fibonacci word and the $(n,i)$-Fibonacci word, satisfy the following properties

\begin{enumerate}[i.]
  \item The word \texttt{\emph{11}} is not a subword of the $i$-Fibonacci word, $i\geq2$.
  \item Let $ab$ be the last two symbols of $f_n^{\left[i\right]}$. For  $n\geq 1$, we have  $ab=\texttt{\emph{10}}$ if  $n$ is even and $ab=\verb"01"$ if $n$ is odd, $i\geq2$.
  \item The concatenation of two successive $i$-Fibonacci words is ``almost commutative'', i.e., $f_{n-1}^{\left[i\right]}f_{n-2}^{\left[i\right]}$ and $f_{n-2}^{\left[i\right]}f_{n-1}^{\left[i\right]}$ have a common prefix of length  $F_{n}^{\left[i\right]}-2$ for all $n\geq2$ and $i\geq 2$.
  \item $\Phi(f_n^{\left[i\right]})$ is a palindrome for all  $n\geq1$.
  \item For all $n\geq6$, $f_n^{\left[i\right]} = f_{n-3}^{\left[i\right]}f_{n-3}^{\left[i\right]}f_{n-6}^{\left[i\right]}l_{n-3}^{\left[i\right]}l_{n-3}^{\left[i\right]}$, where $l_n^{\left[i\right]}=\Phi(f_n^{\left[i\right]})ba$.
\end{enumerate}
\end{proposition}
\begin{proof}
\begin{enumerate}[$i.$]
\item It suffices to prove that $\verb"11"$ is not a subword of  $f_n^{\left[i\right]}$, for $n\geq 0$.  By induction on $n$. For $n=0, 1$ it is clear. Assume for all $j< n$; we prove it for $n$. We know that $f_n^{\left[i\right]}=f_{n-1}^{\left[i\right]}f_{n-2}^{\left[i\right]}$ so by the induction hypothesis we have that $\verb"11"$ is not a subword of $f_{n-1}^{\left[i\right]}$ and $f_{n-2}^{\left[i\right]}$. Therefore, the only possibility is that $\verb"1"$ is a suffix of   $f_{n-1}^{\left[i\right]}$  and $\verb"1"$ is a prefix of   $f_{n-2}^{\left[i\right]}$, but this is impossible.

\item It is clear by induction on $n$.

\item By definition of  $f_{n}^{\left[i\right]}$, we have
  \begin{align*}
   f_{n-1}^{\left[i\right]}f_{n-2}^{\left[i\right]}&=f_{n-2}^{\left[i\right]}f_{n-3}^{\left[i\right]} \cdot f_{n-3}^{\left[i\right]}f_{n-4}^{\left[i\right]}=f_{n-3}^{\left[i\right]}f_{n-4}^{\left[i\right]} \cdot f_{n-3}^{\left[i\right]}f_{n-3}^{\left[i\right]}f_{n-4}^{\left[i\right]},\\
      f_{n-2}^{\left[i\right]}f_{n-1}^{\left[i\right]}&=f_{n-3}^{\left[i\right]}f_{n-4}^{\left[i\right]} \cdot f_{n-2}^{\left[i\right]}f_{n-3}^{\left[i\right]}=f_{n-3}^{\left[i\right]}f_{n-4}^{\left[i\right]} \cdot f_{n-3}^{\left[i\right]}f_{n-4}^{\left[i\right]} \cdot f_{n-3}^{\left[i\right]}.
            \end{align*}
Hence the words have a common prefix of length $F_{n-3}^{\left[i\right]}  + F_{n-4}^{\left[i\right]}  + F_{n-3}^{\left[i\right]} $. By the induction hypothesis $f_{n-3}^{\left[i\right]} f_{n-4}^{\left[i\right]} $ and $f_{n-4}^{\left[i\right]} f_{n-3}^{\left[i\right]} $ have common prefix of length  $F_{n-2}^{\left[i\right]} -2$. Therefore the words have a common prefix of length \[2F_{n-3}^{\left[i\right]}  + F_{n-4}^{\left[i\right]}  + F_{n-2}^{\left[i\right]} -2 = F_{n-2}^{\left[i\right]} +F_{n-1}^{\left[i\right]} -2 = F_{n}^{\left[i\right]} -2. \]

\item By induction on $n$. If $n=2$ then $\Phi(f_2^{\left[i\right]})=\texttt{0}^i$. Now suppose that the result is true for all $j<n$; we prove it for $n$. Then
\begin{align*}
(\Phi(f_{n}^{\left[i\right]}))^R=(\Phi(f_{n-1}^{\left[i\right]}f_{n-2}^{\left[i\right]}))^R=(f_{n-1}^{\left[i\right]}\Phi(f_{n-2}^{\left[i\right]}))^R=\Phi(f_{n-2}^{\left[i\right]})^R(f_{n-1}^{\left[i\right]})^R=\Phi(f_{n-2}^{\left[i\right]})(f_{n-1}^{\left[i\right]})^R.
\end{align*}
If $n$ is even then  $f_n^{\left[i\right]}=\Phi(f_n^{\left[i\right]})\verb"10"$ and
\begin{align*}
\Phi(f_{n}^{\left[i\right]})^R&=\Phi(f_{n-2}^{\left[i\right]})(\Phi(f_{n-1}^{\left[i\right]})\texttt{01})^R =\Phi(f_{n-2}^{\left[i\right]})\texttt{10}\Phi(f_{n-1}^{\left[i\right]})^R=f_{n-2}^{\left[i\right]}\Phi(f_{n-1}^{\left[i\right]})=\Phi(f_{n}^{\left[i\right]}).
\end{align*}
If $n$ is odd, the proof is analogous.

\item By definition of  $f_n^{\left[i\right]}$, we have
\begin{align*}
f_n^{\left[i\right]}&=f_{n-1}^{\left[i\right]}f_{n-2}^{\left[i\right]}=(f_{n-2}^{\left[i\right]}f_{n-3}^{\left[i\right]})(f_{n-3}^{\left[i\right]}f_{n-4}^{\left[i\right]}) \\
&=(f_{n-3}^{\left[i\right]}f_{n-4}^{\left[i\right]})(f_{n-4}^{\left[i\right]}f_{n-5}^{\left[i\right]})f_{n-3}^{\left[i\right]}f_{n-4}^{\left[i\right]} \\
&=f_{n-3}^{\left[i\right]}f_{n-4}^{\left[i\right]}(f_{n-5}^{\left[i\right]}f_{n-6}^{\left[i\right]})f_{n-5}^{\left[i\right]}(f_{n-4}^{\left[i\right]}f_{n-5}^{\left[i\right]})f_{n-4}^{\left[i\right]} \\
&=f_{n-3}^{\left[i\right]}(f_{n-4}^{\left[i\right]}f_{n-5}^{\left[i\right]})f_{n-6}^{\left[i\right]}(f_{n-5}^{\left[i\right]}f_{n-4}^{\left[i\right]})(f_{n-5}^{\left[i\right]}f_{n-4}^{\left[i\right]}) \\
&=f_{n-3}^{\left[i\right]}f_{n-3}^{\left[i\right]}f_{n-6}^{\left[i\right]}l_{n-3}^{\left[i\right]}l_{n-3}^{\left[i\right]}. \qedhere
\end{align*}

\end{enumerate}
\end{proof}

\begin{theorem}\label{slopefibo}
Let $\alpha=\left[0,i, \overline{1}\right]$ be an irrational number, with $i$ a positive integer, then
\begin{align*}
\textbf{w}(\alpha)=\textbf{f}^{\,\left[i\right]}.
\end{align*}
\end{theorem}
\begin{proof}
Let $\alpha=\left[0, i, \overline{1}\right]$ an irrational number, then its associated standard sequence is
\begin{align*}
s_{-1}=\texttt{1}, \ \ s_0=\texttt{0},  \ \ s_1=s_0^{i-1}s_{-1}=\texttt{0}^{i-1}\texttt{1}  \  \text{and}  \    s_n=s_{n-1}s_{n-2}, \ n\geq 2.
\end{align*}
Hence $\left\{s_n\right\}_{n\geq 0} = \left\{f_n^{\left[i\right]}\right\}_{n\geq 0}$ and from Eq. (\ref{fraccon}), we have
 \begin{align*}
\textbf{\emph{w}}(\alpha)&=\lim_{n\rightarrow \infty}s_n=\textbf{\emph{f}}^{\left[i\right]}. \qedhere
\end{align*}
\end{proof}
\textbf{Remark}.  Note that
\begin{align*}
\left[0, i, \overline{1}\right]=\cfrac{1}{i+\cfrac{1}{1+\cfrac{1}{1 + \cfrac{1}{\ddots}}}}=\frac{i-\phi}{i^2-i-1}
\end{align*}
where $\phi$ is the golden ratio.

From the above theorem, we conclude that $i$-Fibonacci words are Sturmian words.
\subsection{The \textbf{$i$}-Fibonacci Word Fractal}

\begin{definition}\label{ifractalfibo}
The $(n,i)$th-curve of Fibonacci, denoted by $\mathcal{F}_n^{\left[i\right]}$, is obtained by applying the odd-even drawing rule to the word $f_n^{\left[i\right]}$. The $i$-Fibonacci word fractal $\mathcal{F}^{\left[i\right]}$ is defined as
\begin{align*}
\mathcal{F}^{\left[i\right]}=\lim_{n\rightarrow \infty}\mathcal{F}_n^{\left[i\right]}.
\end{align*}
\end{definition}

In Table \ref{graf2}, we show the curves $\mathcal{F}_{16}^{\left[i\right]}$ for  $i=1, 2, 3, 4, 5$ and $6$.

\begin{table}[h]
\centering
\begin{tabular}{|c|c|c|} \hline
$\mathcal{F}_{16}^{\left[1\right]}$  & $\mathcal{F}_{16}^{\left[2\right]}$ & $\mathcal{F}_{16}^{\left[3\right]}$  \\
 \includegraphics[scale=0.4]{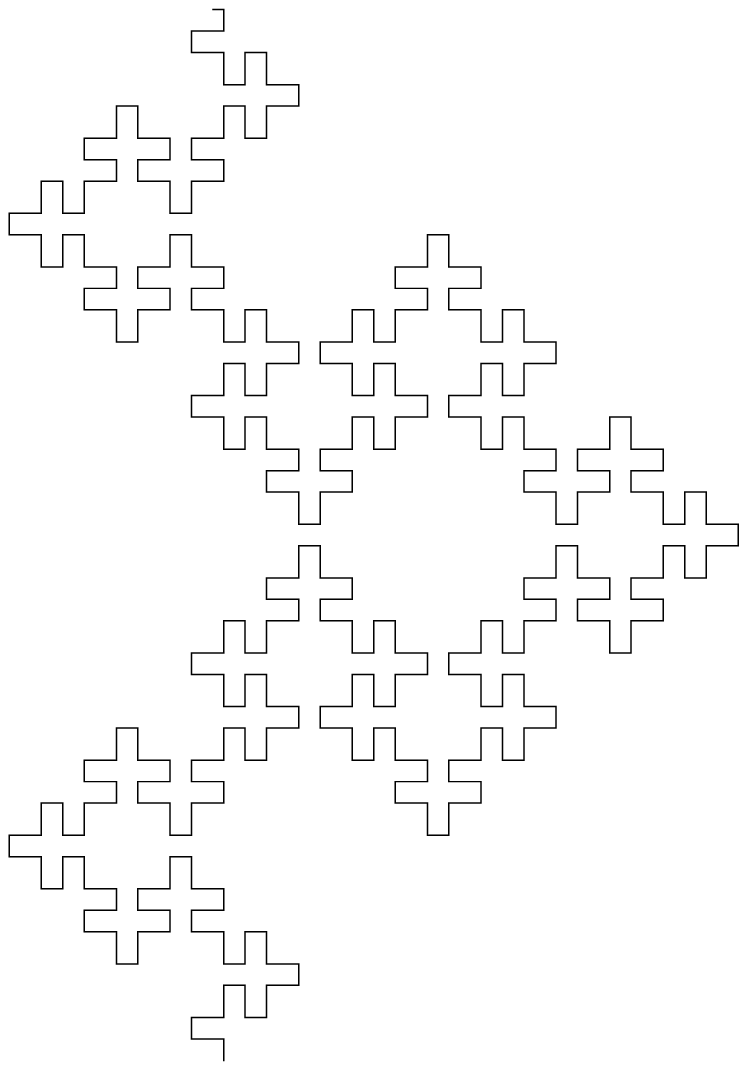} &
    \includegraphics[scale=0.5]{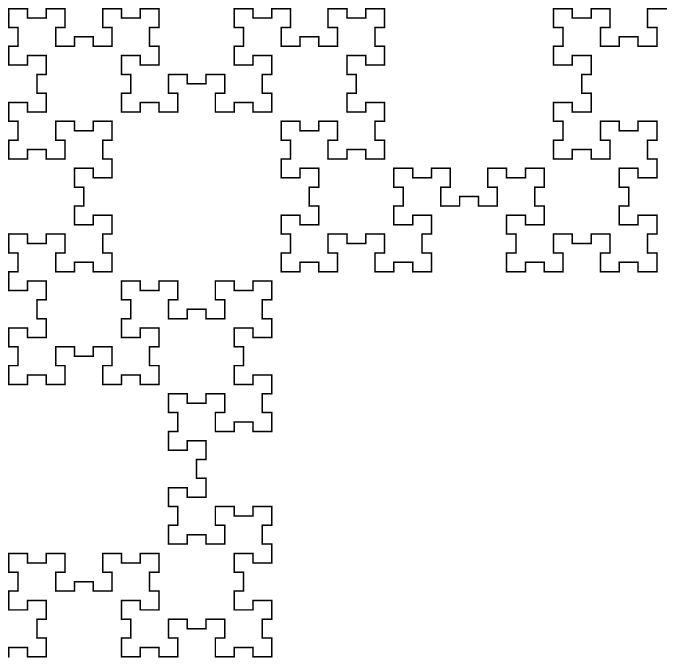} &
      \includegraphics[scale=0.5]{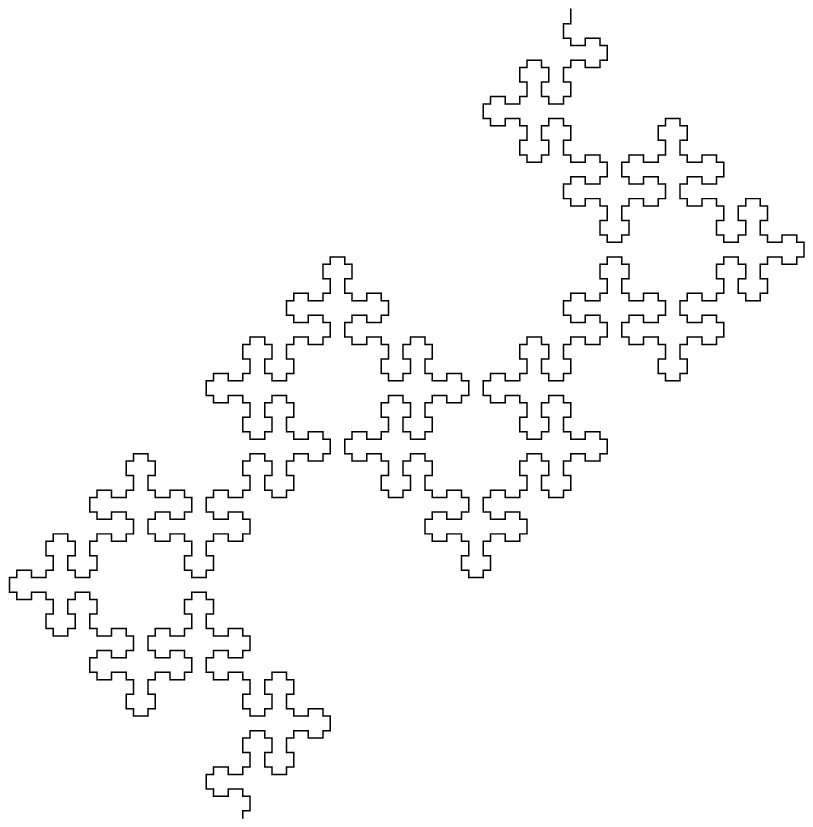} \\ \hline
 $\mathcal{F}_{16}^{\left[4\right]}$ & $\mathcal{F}_{16}^{\left[5\right]}$ & $\mathcal{F}_{16}^{\left[6\right]}$ \\
\includegraphics[scale=0.4]{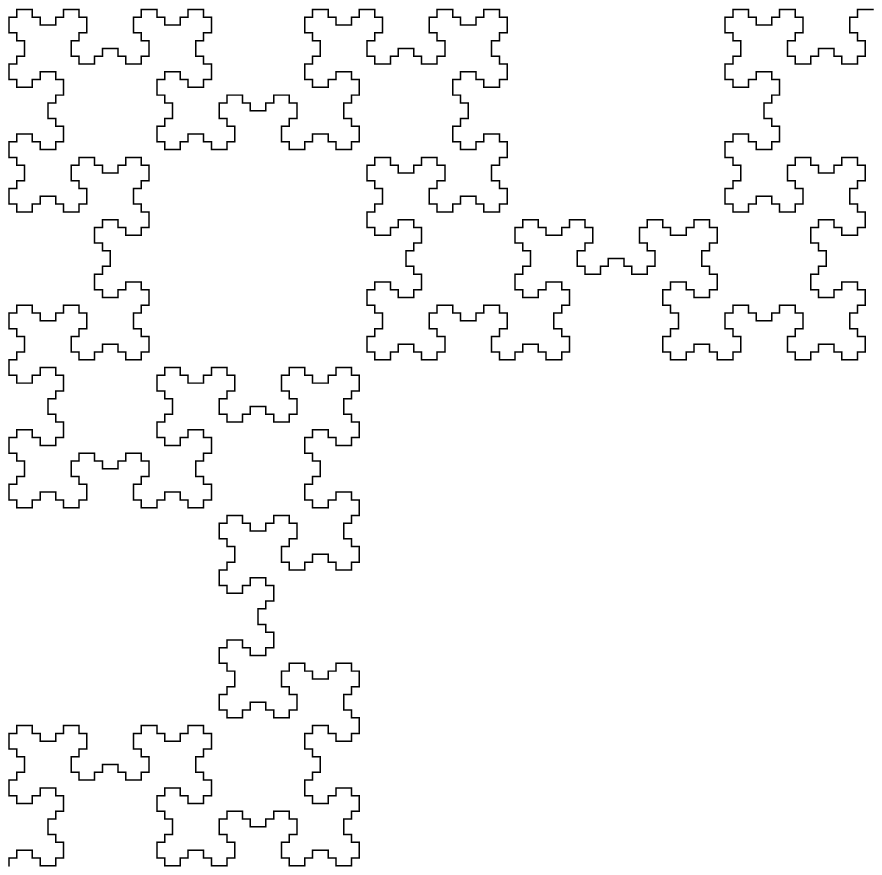} &
     \includegraphics[scale=0.45]{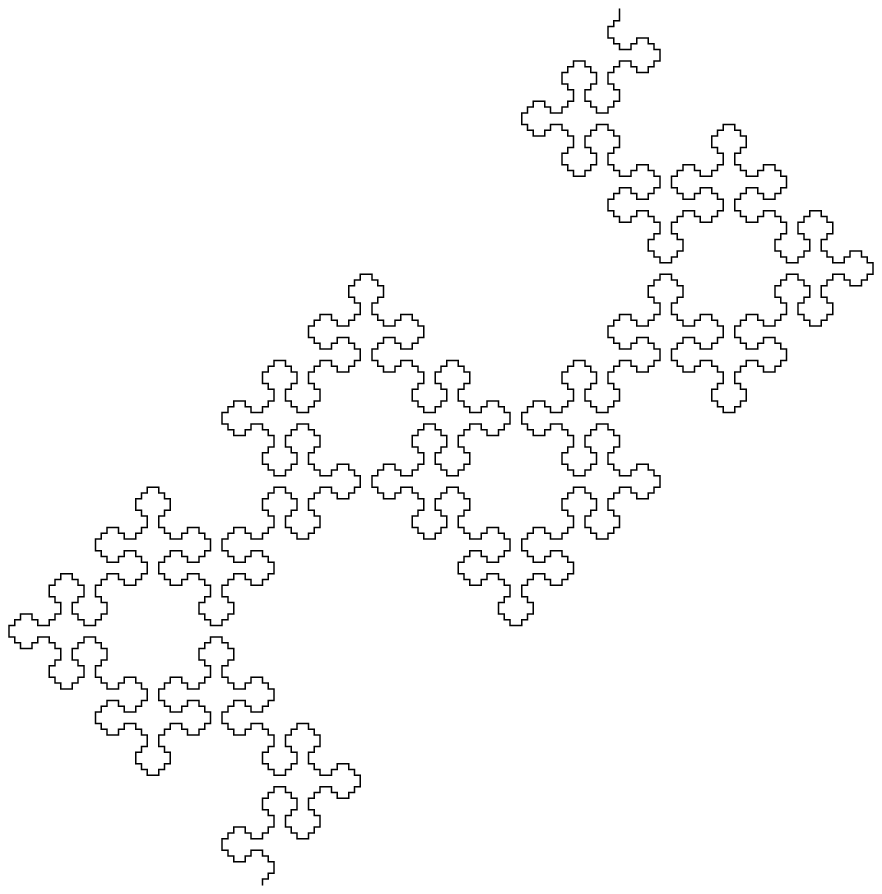} &
    \includegraphics[scale=0.4]{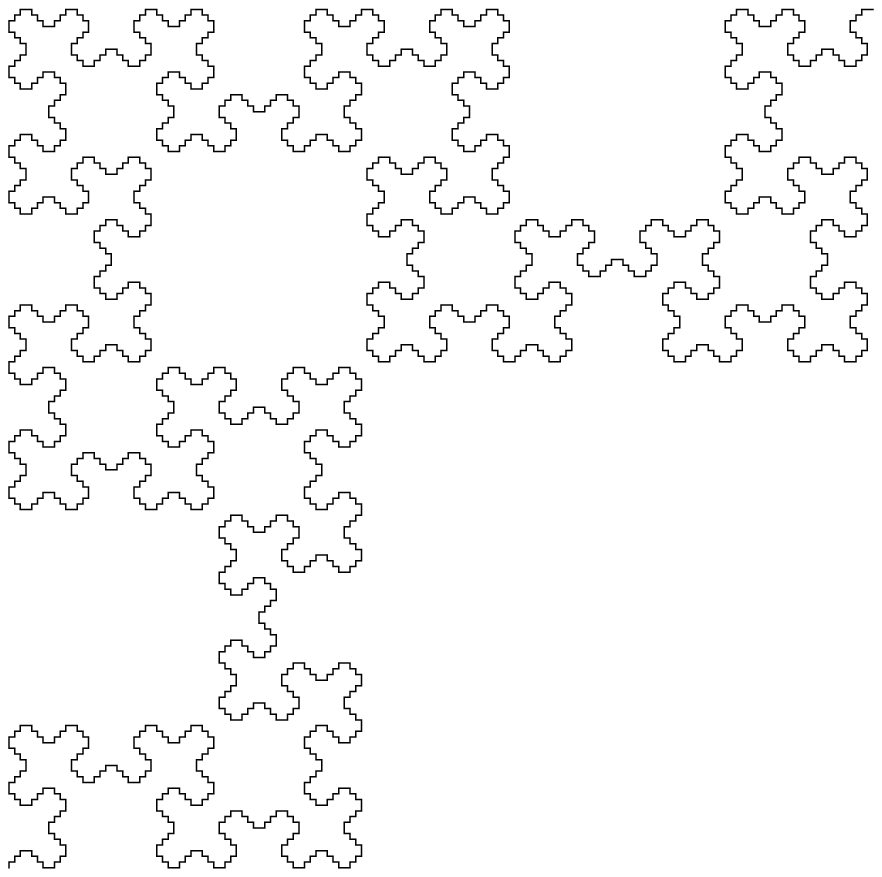} \\ \hline
\end{tabular}
\caption{Curves $\mathcal{F}_{16}^{\left[i\right]}$ for  $i=1, 2, 3, 4, 5$ and $6$.}
\label{graf2}
\end{table}

The following proposition generalizes  Proposition \ref{fractalfibo}.

\begin{proposition}\label{ifibofractal} The $i$-Fibonacci word fractal  and the curve $\mathcal{F}_n^{\left[i\right]}$  have the following properties:
\begin{enumerate}[i.]
  \item The Fibonacci fractal $\mathcal{F}^{\left[i\right]}$ is composed only of segments of length 1 or 2.
  \item The curve $\mathcal{F}_n^{\left[i\right]}$ is similar to the curve $\mathcal{F}_{n-3}^{\left[i\right]}$.
  \item The curve $\mathcal{F}_n^{\left[i\right]}$ is composed  of 5 curves: $\mathcal{F}_{n}^{\left[i\right]}=\mathcal{F}_{n-3}^{\left[i\right]}\mathcal{F}_{n-3}^{\left[i\right]}\mathcal{F}_{n-6}^{\left[i\right]}
      \mathcal{F'}_{n-3}^{\left[i\right]}\mathcal{F'}_{n-3}^{\left[i\right]}$.
  \item The  curve $\mathcal{F}_{n}^{\left[i\right]}$ is symmetric. More precisely, the curves $\mathcal{F}_{3n}^{\left[i\right]}$ and $\mathcal{F}_{3n+2}^{\left[i\right]}$ are symmetric with respect to a line and the curve $\mathcal{F}_{3n+1}^{\left[i\right]}$ is symmetric with respect to a point.

  \item The scale factor between  $\mathcal{F}_n^{\left[i\right]}$  and $\mathcal{F}_{n-3}^{\left[i\right]}$ is $1+\sqrt{2}$.
\end{enumerate}
\end{proposition}
\begin{proof}
\begin{enumerate}[$i.$]
\item It is clear from Proposition \ref{fibog}-$i$, because \texttt{110} and \texttt{111} are not subwords of  $\textbf{\emph{f}}^{\left[i\right]}$.
\item By Proposition \ref{morfismo} we have $\emph{f}_{n-1}^{\,\left[i+2\right]}=\varphi_i(f_n)$ for all integer $n\geq2$ and $i\geq 0$. Moreover,  $\varphi_i$ maps the different segments as shown in Table  \ref{codigofib}.

 \begin{table}[h]
\centering
\begin{tabular}{|c|c|c|}
  \hline
  \multicolumn{3}{|c|}{If $i$ is even}\\ \hline
$\varphi_i(\verb"01")=\texttt{0}^{i+1}\texttt{1}$  &  $\varphi_i(\texttt{10})=\texttt{0}^{i}\texttt{10}$ &
$\varphi_i(\texttt{00})=\texttt{00}$ \\
 \psset{unit=3mm}
\begin{pspicture}(0,-2)(13,5)
\psline[linewidth=1pt]{-}(0,0)(0,3)(3,3)
\psline[linewidth=1pt]{->}(3,3)(4,3)
\rput(6,2){$\longmapsto$}
\psline[linewidth=1pt]{-}(8,0)(8,1)(9,1)(9,2)(10,2)
\psline[linestyle=dotted, linewidth=1pt]{-}(10,2)(11,2)(11,3)(12,3)
\psline[linewidth=1pt]{-}(12,3)(13,3)
\psline[linewidth=1pt]{->}(13,3)(13.2,3)
\rput(10.2,-0.3){$\underbrace{\hspace{1.4cm}}$}
\rput(10.2,-1.5){$\frac{i+2}{2}$}
\end{pspicture}

   & \psset{unit=3mm}
\begin{pspicture}(0,-1)(11,7)
\psline[linewidth=1pt]{-}(2,0)(2,3)(2,6)
\psline[linewidth=1pt]{->}(2,6)(1,6)
\rput(4,3){$\longmapsto$}
\psline[linewidth=1pt]{-}(6,1)(6,2)(7,2)(7,3)(8,3)
\psline[linestyle=dotted, linewidth=1pt]{-}(8,3)(9,3)(9,4)
\psline[linewidth=1pt]{-}(9,4)(10,4)
\psline[linewidth=1pt]{-}(10,4)(10,6)
\psline[linewidth=1pt]{->}(10,6)(9.4,6)
\rput(8.2,0.7){$\underbrace{\hspace{1.4cm}}$}
\rput(8.2,-0.5){$\frac{i}{2}$}

\end{pspicture}
 &  \psset{unit=3mm}
\begin{pspicture}(0,-2)(10,4)
\psline[linewidth=1pt]{-}(0,0)(0,3)(3,3)
\psline[linewidth=1pt]{->}(3,3)(3,4)
\rput(5,2){$\longmapsto$}
\psline[linewidth=1pt]{-}(7,0)(7,3)(10,3)
\psline[linewidth=1pt]{->}(10,3)(10,4)
\end{pspicture}
\\
  \hline
    \multicolumn{3}{|c|}{If $i$ is odd}\\ \hline
$\varphi_i(\verb"01")=\texttt{0}^{i+1}\texttt{1}$  &  $\varphi_i(\texttt{10})=\texttt{0}^{i}\texttt{10}$ &
$\varphi_i(\texttt{00})=\texttt{00}$ \\
 \psset{unit=3mm}
\begin{pspicture}(0,-2)(12,5)
\psline[linewidth=1pt]{-}(0,0)(0,3)(3,3)
\psline[linewidth=1pt]{->}(3,3)(4,3)
\rput(6,2){$\longmapsto$}
\psline[linewidth=1pt]{-}(8,0)(8,1)(9,1)(9,2)(10,2)
\psline[linestyle=dotted, linewidth=1pt]{-}(10,2)(11,2)(11,3)(12,3)
\psline[linewidth=1pt]{-}(12,3)(12,4)
\psline[linewidth=1pt]{->}(12,4)(12,4.2)
\rput(10.1,-0.3){$\underbrace{\hspace{1.3cm}}$}
\rput(10.2,-1.5){$\frac{i+1}{2}$}
\end{pspicture}

   & \psset{unit=3mm}
\begin{pspicture}(0,-1)(11,7)
\psline[linewidth=1pt]{-}(2,0)(2,3)(2,6)
\psline[linewidth=1pt]{->}(2,6)(1,6)
\rput(4,3){$\longmapsto$}
\psline[linewidth=1pt]{-}(6,1)(6,2)(7,2)(7,3)(8,3)
\psline[linestyle=dotted, linewidth=1pt]{-}(8,3)(9,3)(9,4)
\psline[linewidth=1pt]{-}(9,4)(11,4)
\psline[linewidth=1pt]{->}(11,4)(11,3.5)
\rput(7.9,0.7){$\underbrace{\hspace{1.2cm}}$}
\rput(8.1,-0.5){$\frac{i+1}{2}$}

\end{pspicture}
 &  \psset{unit=3mm}
\begin{pspicture}(0,-2)(10,4)
\psline[linewidth=1pt]{-}(0,0)(0,3)(3,3)
\psline[linewidth=1pt]{->}(3,3)(3,4)
\rput(5,2){$\longmapsto$}
\psline[linewidth=1pt]{-}(7,0)(7,3)(10,3)
\psline[linewidth=1pt]{->}(10,3)(10,4)
\end{pspicture}
\\  \hline
\end{tabular}
\caption{Mapping of segments.}
\label{codigofib}
\end{table}
For example in Fig. \ref{grafunion}, we show the mapping of $f_{10}$ by $\varphi_i$ when $i=2, 3$.
\begin{figure}[H]
\centering
  \includegraphics[scale=0.4]{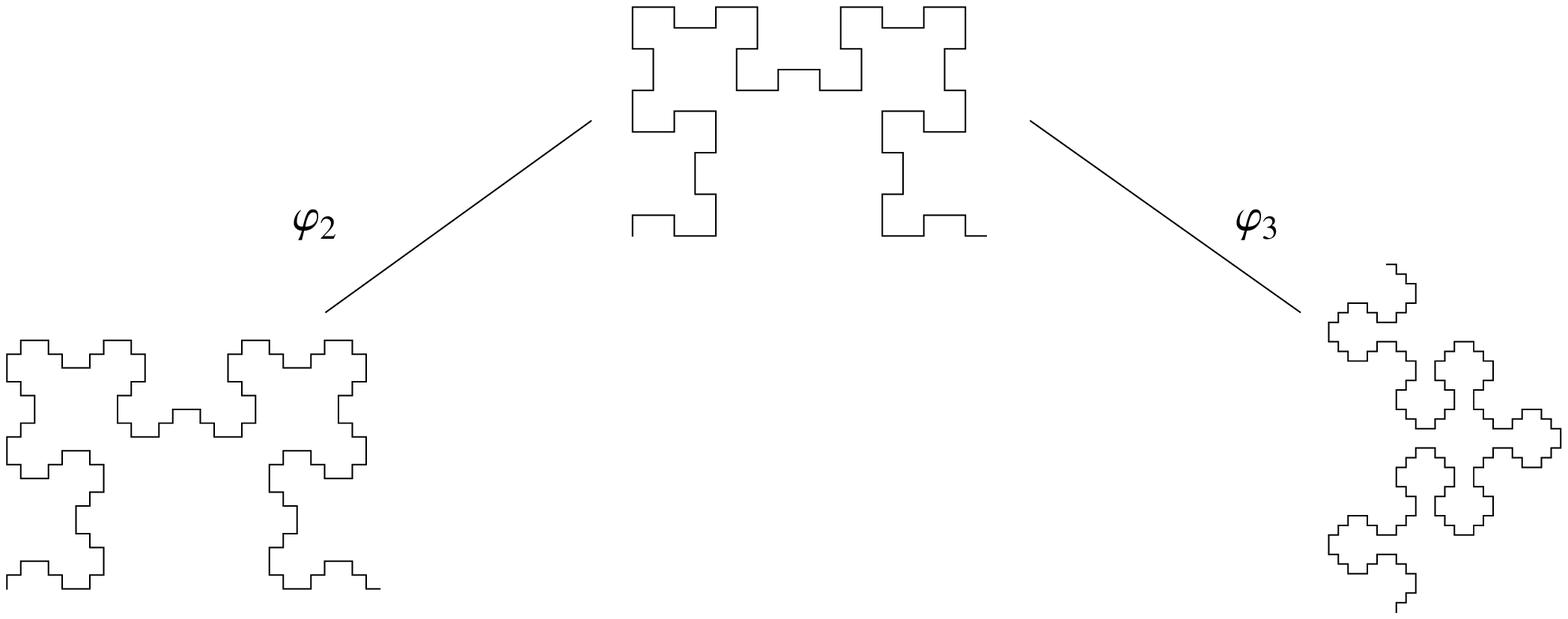}
  \caption{Mapping of $\varphi_2(f_{10})$ and $\varphi_3(f_{10})$. } \label{grafunion}
\end{figure}
Hence, it is clear that $\varphi_i$  preserves the geometric properties.  By
Proposition  \ref{fractalfibo} we have  $\mathcal{F}_n$  is similar to the curve $\mathcal{F}_{n-3}$     then  $\mathcal{F}_n^{\left[i\right]}$  is similar to  $\mathcal{F}_{n-3}^{\left[i\right]}$.

\item It is clear from  Proposition \ref{fibog}-$v$.

\item The proof runs like in $ii$.

\item We show that $$f_n^{\left[i\right]} = f_{n-3}^{\left[i\right]}f_{n-3}^{\left[i\right]}f_{n-6}^{\left[i\right]}l_{n-3}^{\left[i\right]}l_{n-3}^{\left[i\right]}=\Phi(f_{n-3}^{\left[i\right]})ab\Phi(f_{n-3}^{\left[i\right]})ab f_{n-6}^{\left[i\right]}\Phi(l_{n-3}^{\left[i\right]})ba\Phi(l_{n-3}^{\left[i\right]})ba.$$ Since  $ab$ is either \texttt{01} or \texttt{10} , and $\mathcal{F}_{n}^{\left[i\right]}=\mathcal{F}_{n-3}^{\left[i\right]}\mathcal{F}_{n-3}^{\left[i\right]}\mathcal{F}_{n-6}^{\left[i\right]}
      \mathcal{F'}_{n-3}^{\left[i\right]}\mathcal{F'}_{n-3}^{\left[i\right]}$, then the first two curves are orthogonal and the last two curves are orthogonal. Let $L_n^{\left[i\right]}$ be the length of the curve $\mathcal{F}_n^{\left[i\right]}$ from first to last point drawn. Then $L_n^{\left[i\right]}=2L_{n-3}^{\left[i\right]}+L_{n-6}^{\left[i\right]}$ and  by  definition, the scale factor $\Gamma$ is
   \begin{align*}
\Gamma=\frac{L_n^{\left[i\right]}}{L_{n-3}^{\left[i\right]}}=\frac{L_{n-3}^{\left[i\right]}}{L_{n-6}^{\left[i\right]}}
\end{align*}
hence $\Gamma L_{n-3}^{\left[i\right]}=L_n^{\left[i\right]}=2L_{n-3}^{\left[i\right]}+L_{n-6}^{\left[i\right]}=2L_{n-3}^{\left[i\right]}+\frac{L_{n-3}^{\left[i\right]}}{\Gamma}$,  then $\Gamma=1+\sqrt2$. \qedhere
\end{enumerate}
\end{proof}
For each $i$ the system $\mathcal{F}_n^{\left[i\right]}$ ($n \geq 0$) has as attractor the curve
$\mathcal{F}$ (the same argument given in Proposition \ref{ifibofractal}-$ii$).

\section{Generalized Fibonacci Snowflakes}
 We say that a path $w$ is \emph{closed} if it satisfies $|w|_0=|w|_2$ and $|w|_1=|w|_3$.  A \emph{simple path} is a word $w$ such that none of its proper subwords is a closed path. A \emph{boundary word} is a closed path such that none of its proper subwords is closed. Therefore, a \emph{polyomino} is a subset of $\mathbb{Z}\times\mathbb{Z}$ contained in some boundary word.

\begin{figure}[h]
\begin{minipage}[b]{0.55\linewidth}
\begin{example}
In Fig.\ref{pol1} we show a polyomino $P$ such that starting from point $S$, (counterclockwise) the boundary $\textbf{b}(P)$  is coded by the word $w=\verb"2122323030103011"$. Moreover, we denoted by $\widehat{w}$ the path traveled in the opposite direction, i.e.,  $\widehat{w}=\rho^2(w^R)$, where $\rho^2$ is the morphism defined by  $\rho^2(a)=2+ a$, $a\in \mathcal{A}$. In this example $\widehat{w}=\rho^2(\verb"1103010303232212")=\verb"3321232121010030".$
\end{example}

 \end{minipage}
\hspace{0.1cm}
\begin{minipage}[b]{0.45\linewidth}
\centering
\includegraphics[scale=0.4]{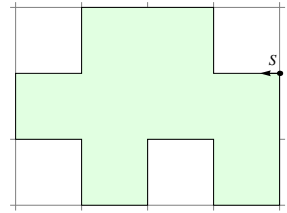}
\caption{Polyomino $P$.}
\label{pol1}
\end{minipage}
\end{figure}

In this section, we study a new generalization of Fibonacci polyominoes from  $i$-Fibonacci words. We use the same procedure as in \cite{BLO2} and we present some geometric properties.

\subsection{Construction of Generalized Fibonacci Polyominoes}

First, rewrite the $i$-Fibonacci words over alphabet $\left\{\verb"0", \verb"2"\right\}\subset \mathcal{A}$,  specifically we apply the morphism $\verb"0" \rightarrow  \verb"2 ", \verb"1" \rightarrow  \verb"0"$.
Next, apply the operator $\Sigma_1$ followed by the operator $\Sigma_0$, where
\begin{align*}
\Sigma_{\alpha}(w)=\alpha\cdot(\alpha + w_1)\cdot(\alpha + w_1 + w_2)\cdots(\alpha + w_1 + w_2 + \cdots w_n),
\end{align*}
with $\alpha\in \mathcal{A}$ and $w=w_1w_2\cdots w_n$. This yield the words $\textbf{p}^{\left[i\right]}=\Sigma_0\Sigma_1\textbf{\emph{f}}^{\left[i\right]}$.

\begin{example}
In Table \ref{graf3}, we show the first words  $\textbf{\emph{p}}^{\left[i\right]}$, with its corresponding curves. The case $n=2$ corresponds to a version of the Fibonacci word fractal with only segments of length 1 \emph{\cite{BLO2}}.
\begin{table}[h]
\centering
\begin{tabular}{|c|c|c|} \hline
$\textbf{p}^{\left[1\right]}=\textbf{p}^{\left[2\right]}=\verb"010303230301"\cdots$   &  $\textbf{p}^{\left[3\right]}=\verb"01012121010303"\cdots$  &  $\textbf{p}^{\left[4\right]}=\verb"01010303032323"\cdots$ \\
 \includegraphics[scale=0.55]{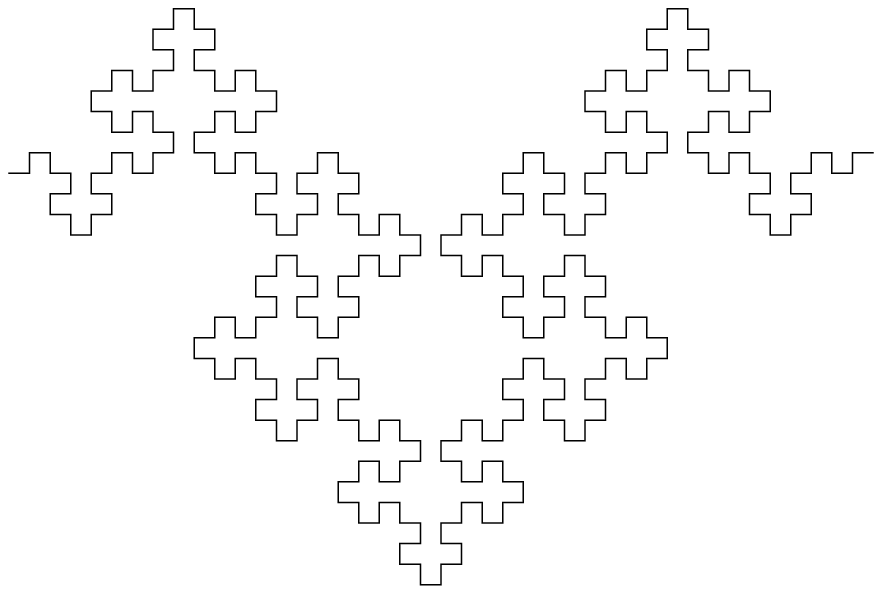} &
      \includegraphics[scale=0.45]{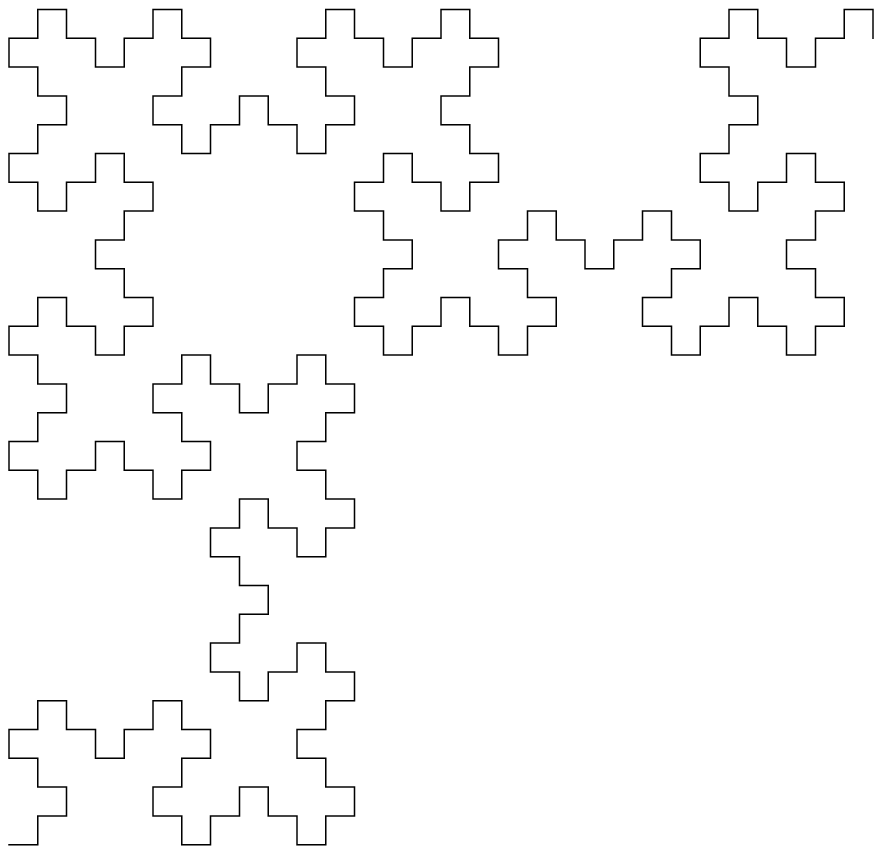} &  \includegraphics[scale=0.45]{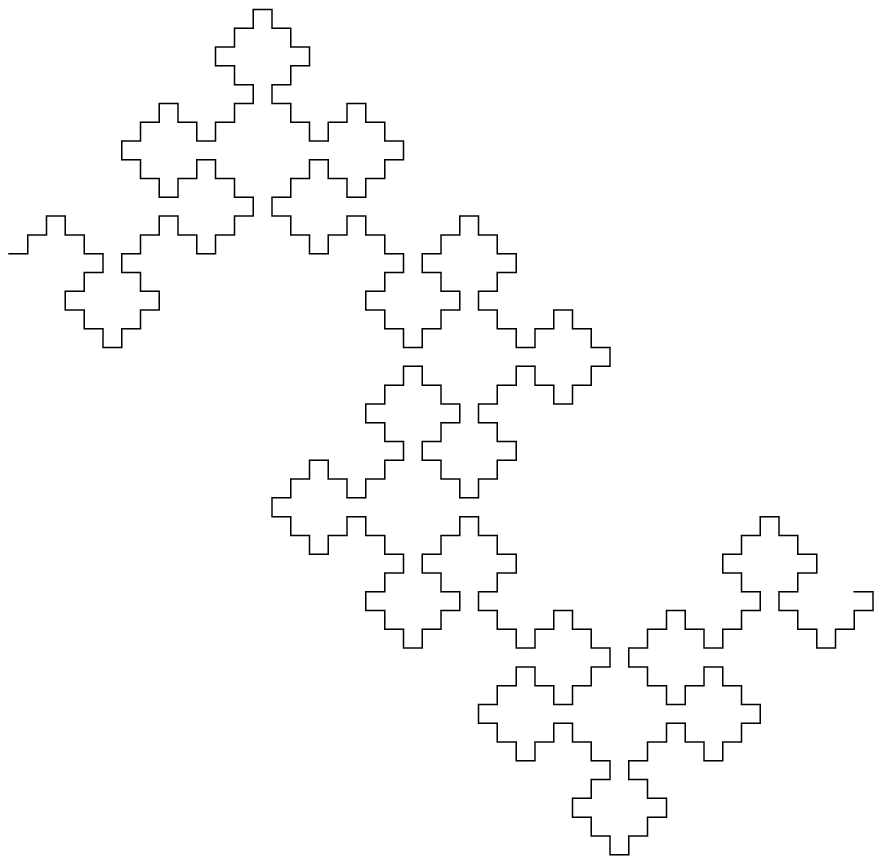} \\ \hline
   $\textbf{p}^{\left[5\right]}=\verb"01010121212101"\cdots$ & $\textbf{p}^{\left[6\right]}=\verb" 01010103030303 "\cdots$ & $\textbf{p}^{\left[7\right]}=\verb"01010101212121"\cdots$  \\

     \includegraphics[scale=0.43]{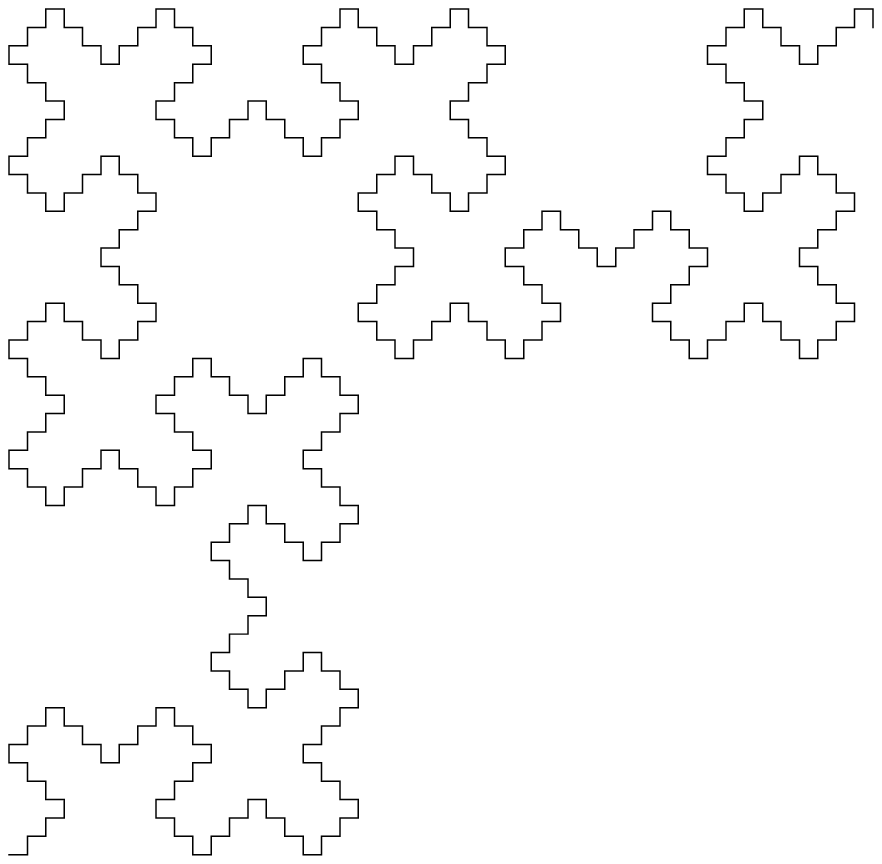} & \includegraphics[scale=0.5]{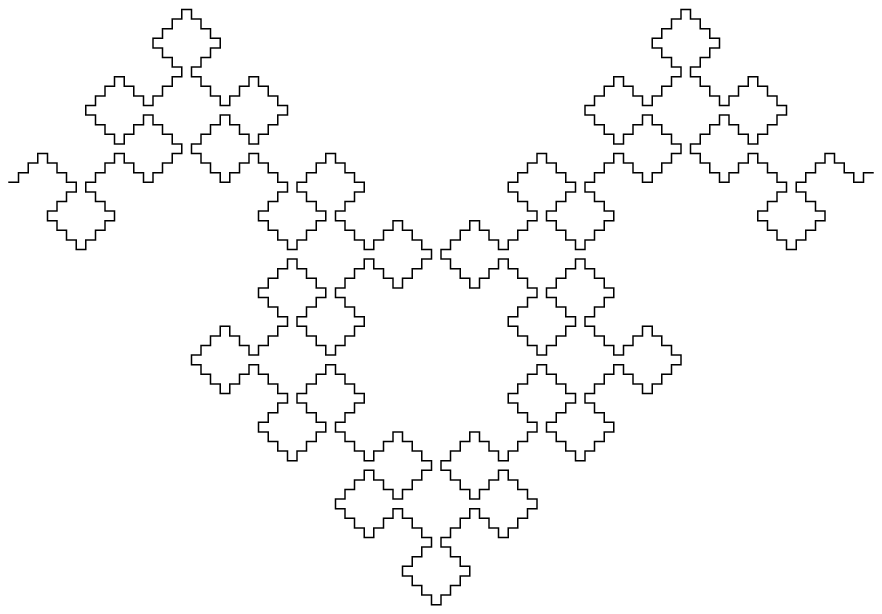} & \includegraphics[scale=0.43]{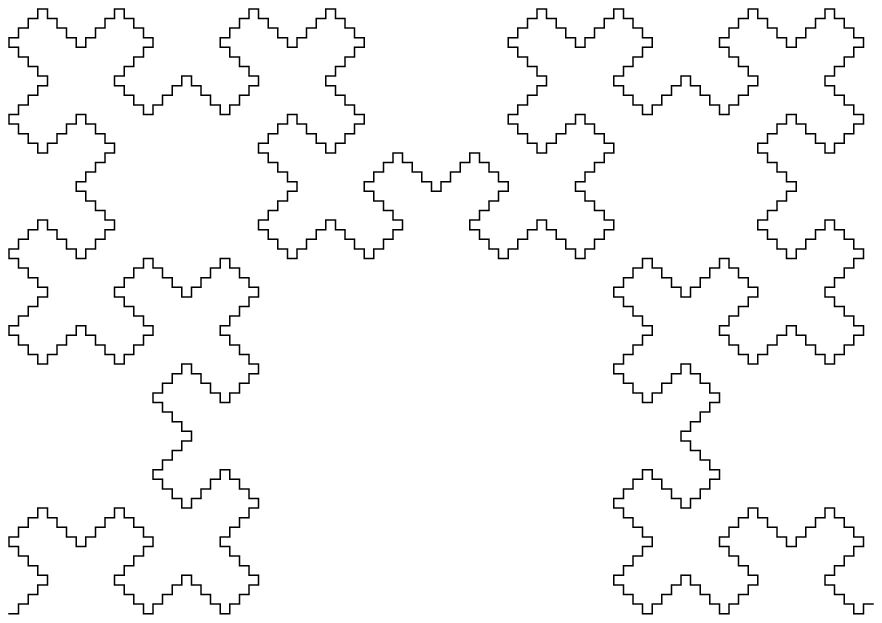} \\ \hline
\end{tabular}
\caption{Words $\textbf{p}^{\left[i\right]}$ and its corresponding curves.}
\label{graf3}
\end{table}
\end{example}

Given a word $w\in\mathcal{A}^*$ we  define the word  $\Delta(w)=(w_2 - w_1)\cdot(w_3 - w_2)\cdots (w_n - w_{n-1})\in \mathcal{A}^*$, then it is clear that  $\Delta (\textbf{p}^{\left[i\right]})=\Sigma_1\textbf{\emph{f}}^{\left[i\right]}$. We shall denote this sequence  by   $\textbf{q}^{\left[i\right]}$. Last, we define the morphism $\overline{a}$, with $a\in\mathcal{A}$, as $
\overline{\texttt{0}}=\verb"0", \overline{\texttt{1}}=\verb"3",  \overline{\texttt{2}}=\verb"2",  \overline{\texttt{3}}=\verb"1".$ Moreover, the words $w\in \mathcal{A}^*$ satisfying $\overline{w}=w^R$ are called\emph{ antipalindromes}.

\begin{definition}
Consider the sequence  $\left\{q_n^{\left[i\right]}\right\}_{n\geq 0}$ defined by:
\begin{itemize}
  \item If $i$ is even, $q_0^{\left[i\right]}=\epsilon$, $q_1^{\left[i\right]}=\verb"1"$, $q_2^{\left[i\right]}=(\verb"13")^{\frac{i}{2}}$ and
\begin{align*}
q_n^{\left[i\right]}=\begin{cases}
q_{n-1}^{\left[i\right]}q_{n-2}^{\left[i\right]}, & \ n\cong 1 \mod3\\
q_{n-1}^{\left[i\right]}\overline{q_{n-2}^{\left[i\right]}}, & \ n\cong 0, 2 \mod3.
\end{cases}
\end{align*}
  \item If $i$ is odd, $q_0^{\left[i\right]}=\epsilon$, $q_1^{\left[i\right]}\verb"=1"$, $q_2^{\left[i\right]}=(\verb"13")^{\frac{i-1}{2}}\verb"1"$ and
\begin{align*}
q_n^{\left[i\right]}=
\begin{cases}
q_{n-1}^{\left[i\right]}q_{n-2}^{\left[i\right]}, & \ n\cong 0 \mod3\\
q_{n-1}^{\left[i\right]}\overline{q_{n-2}^{\left[i\right]}}, & \ n\cong 1, 2 \mod3.
\end{cases}
\end{align*}
\end{itemize}
\end{definition}
It is clear that $|q_n^{\left[i\right]}|=F_{n-1}^{\left[i\right]}$.

\begin{example} The first terms of   $\left\{q_n^{\left[i\right]}\right\}_{n\geq 0}$ are:

\begin{tabular}{l}
  $\left\{q_n^{\left[2\right]}\right\}_{n\geq 0}=\left\{\epsilon \verb", 1, 13, 133, 13313, 13313311, 1331331131131," \ldots\right\}$, \\
  $\left\{q_n^{\left[3\right]}\right\}_{n\geq 0}=\left\{\epsilon\verb", 1, 131, 1311, 1311313, 13113133133, 131131331331311313,"\ldots \ldots\right\}$, \\
  $\left\{q_n^{\left[4\right]}\right\}_{n\geq 0}=\left\{\epsilon\verb", 1, 1313, 13133, 131331313, 13133131331311," \ldots\right\}$, \\
    $\left\{q_n^{\left[5\right]}\right\}_{n\geq 0}=\left\{\epsilon\verb", 1, 13131, 131311, 13131131313, 13131131313313133," \ldots\right\}$.
\end{tabular}
\end{example}

The following propositions generalize the case when $i=2$, \cite{BLO2}.
\begin{proposition}
The word  $\textbf{\emph{q}}^{\left[i\right]}=\Sigma_1\textbf{f}^{\left[\,i\right]}$ is the limit of the sequence $\left\{q_n^{\left[i\right]}\right\}_{n\geq0}$.
\end{proposition}

\begin{proof}
We know that $\Delta(\textbf{q}^{\left[i\right]})=\textbf{\emph{f}}^{\left[i\right]}$, then it suffices to prove that $\Delta(q_n^{\left[i\right]})\alpha_{n-1}=f_{n-1}^{\left[i\right]}$ for all $n\geq 2$, where $\alpha_n=\texttt{2}$ if $n$ is even and  $\alpha_n=\texttt{0}$ if $n$ is odd. By induction on $n$. If $i$ is even, then
\begin{align*}
\Delta(q_2^{\left[i\right]})\alpha_{1}&=\Delta((\texttt{13})^{i/2})\alpha_1=(\texttt{22})^{i/2-1}\texttt{20}= \texttt{2}^{i-2}\texttt{20}=\texttt{2}^{i-1}\texttt{0}=f_1^{\left[i\right]},\\
\Delta(q_3^{\left[i\right]})\alpha_{2}&=\Delta((\texttt{13})^{i/2}\texttt{3})\alpha_2=\texttt{2}^{i-1}\texttt{02}= f_1^{\left[i\right]}f_0^{\left[i\right]}=f_2^{\left[i\right]},\\
\Delta(q_4^{\left[i\right]})\alpha_{3}&=\Delta((\texttt{13})^{i/2}\texttt{3}(\texttt{13})^{i/2})\alpha_3= \texttt{2}^{i-1}\texttt{0}\texttt{2}^i\texttt{0} =f_2^{\left[i\right]}f_1^{\left[i\right]}=f_3^{\left[i\right]}.
\end{align*}

Assume for all $m$, with $2\leq m < n$; we prove it for $n$. We only prove the case  $n\cong 1 \mod 3$, since the argument is similar for the other cases. Let $n=3k+1$ for some integer $k$. Then
\begin{align*}
\Delta(q_{3k+1}^{\left[i\right]})\alpha_{3k}&=\Delta(q_{3k}^{\left[i\right]}q_{3k-1}^{\left[i\right]})\alpha_{3k}=\Delta(q_{3k}^{\left[i\right]})\alpha_{3k-1}\Delta(q_{3k-1}^{\left[i\right]})\alpha_{3k-2}
=f_{3k-1}^{\left[i\right]}f_{3k-2}^{\left[i\right]}= f_{3k}^{\left[i\right]}.
\end{align*}
If $i$ is odd, the proof is similar.
\end{proof}

\begin{proposition}\label{pal}
Let $n\in \mathbb{N}$ and $\sigma_n=\texttt{\emph{1}}$ if  $n$ is even and $\sigma_n=\texttt{\emph{3}}$ if $n$ is odd. Then if $i$ is even $q_{3n+1}^{\left[i\right]}=r\sigma_n, q_{3n+2}^{\left[i\right]}=m\overline{\sigma_n}$ and $q_{3n+3}^{\left[i\right]}=p\overline{\sigma_n}$ for some  antipalindrome $p$  and  some palindromes  $r$, $m$.  If $i$ is odd $q_{3n+1}^{\left[i\right]}=r\overline{\sigma_n}, q_{3n+2}^{\left[i\right]}=m\overline{\sigma_n}$ and $q_{3n+3}^{\left[i\right]}=p\sigma_n$ for some  antipalindrome $m$ and  some palindromes $r$,  $p$.
\end{proposition}
\begin{proof}
The proof is by induction on $n$. If $i$ is even, for $n=0$ we have  $q_1^{\left[i\right]}=\epsilon \cdot \texttt{1}, q_2^{\left[i\right]}=(\texttt{13})^{i/2}=((\texttt{13})^{\frac{i}{2}-1}\texttt{1})\texttt{3}=((\texttt{13})^{\frac{i}{2}-1}\texttt{1})
\cdot\overline{\texttt{1}}$ and $q_3^{\left[i\right]}=(\texttt{13})^{i/2}\texttt{3}=(\texttt{13})^{i/2}\cdot\overline{\texttt{1}}$. Now, suppose that  $q_{3n+1}=r\sigma_n, q_{3n+2}=m\overline{\sigma_n}$ and $q_{3n+3}=p\overline{\sigma_n}$ for some  antipalindrome $p$ and some palindromes  $r$, $m$.  Then
\begin{align*}
q_{3n+4}^{\left[i\right]}&=q_{3n+3}^{\left[i\right]}q_{3n+2}^{\left[i\right]}=q_{3n+2}^{\left[i\right]}\overline{q_{3n+1}^{\left[i\right]}}q_{3n+2}^{\left[i\right]}=m\overline{\sigma_n} \cdot\overline{r \sigma_n} \cdot m \overline{\sigma_n}=m\overline{\sigma_n r \sigma_n} m\cdot \sigma_{n+1},  \\
q_{3n+5}^{\left[i\right]}&=q_{3n+4}^{\left[i\right]}\overline{q_{3n+3}^{\left[i\right]}}=q_{3n+3}^{\left[i\right]}q_{3n+2}^{\left[i\right]}\overline{q_{3n+3}^{\left[i\right]}}=p\overline{\sigma_n } \cdot m \overline{\sigma_n} \cdot \overline {p  \overline{\sigma_n}}=p\overline{\sigma_n }m \overline{\sigma_n p}\cdot \overline{\sigma_{n+1}},\\
q_{3n+6}^{\left[i\right]}&=q_{3n+5}^{\left[i\right]}\overline{q_{3n+4}^{\left[i\right]}}=
q_{3n+4}^{\left[i\right]}\overline{q_{3n+3}^{\left[i\right]}}\overline{q_{3n+4}^{\left[i\right]}}=
m\overline{\sigma_n r \sigma_n}m \overline{\sigma_n} \cdot \overline{p \sigma_{n+1}}\cdot  \overline{m}  \sigma_n r \sigma_n \overline{m} \sigma_n \\
&=m\overline{\sigma_n r \sigma_n}m \overline{\sigma_n p} \sigma_n \overline{m}  \sigma_n r \sigma_n \overline{m} \cdot  \overline{\sigma_{n+1}}
\end{align*}
with palindromes  $m\overline{\sigma_n r \sigma_n}m $ and $p\overline{\sigma_n }m \overline{\sigma_n p}$, and  antipalindrome $m\overline{\sigma_n r \sigma_n}m \overline{\sigma_n p} \sigma_n \overline{m}  \sigma_n r \sigma_n \overline{m}$.
If $i$ is odd, the proof is similar.
\end{proof}

\begin{proposition}\label{poliminofibo}
Let $n$ be a positive integer  and $\alpha\in \mathcal{A}$ then
\begin{enumerate}[i.]
  \item The path  $\Sigma_\alpha q_n^{\left[i\right]}$ is simple.
  \item If $i$ is even, then the path $\Sigma_\alpha^{\circ}(q_{3n}^{\left[i\right]})^4$ is the boundary word of a polyomino.
\item If $i$ is odd, then the path $\Sigma_\alpha^{\circ}(q_{3n+2}^{\left[i\right]})^4$ is the boundary word of a polyomino.
\end{enumerate}
Where $\Sigma^{\circ}_{\alpha}(w)=\alpha\cdot(\alpha + w_1)\cdot(\alpha + w_1 + w_2)\cdots(\alpha + w_1 + w_2 + \cdots w_{n-1})$.
\end{proposition}
\begin{proof} \begin{enumerate}[$i.$]
\item The proof is by induction on $n$. It is the similar to   \cite{BLO2} or \cite{FIB}, we only describe the basic ideas because the proof is rather technical. For $n=1, 2, 3$ it is clear. Assume for all $j$ such that $1\leqslant j < n$; we prove it for $n$. The idea is to divide the path $\Sigma_\alpha q_n^{\left[i\right]}$ into three smaller parts, for example
the path $\Sigma_0 q_{12}^{\left[5\right]}$ is divided into parts $\Sigma_0 q_{10}^{\left[5\right]}$, $\Sigma_2 q_{9}^{\left[5\right]}$ and $\Sigma_3 q_{10}^{\left[5\right]}$, (see Fig. \ref{divide}).

\begin{figure}[!h]
\centering
 \includegraphics[scale=0.4]{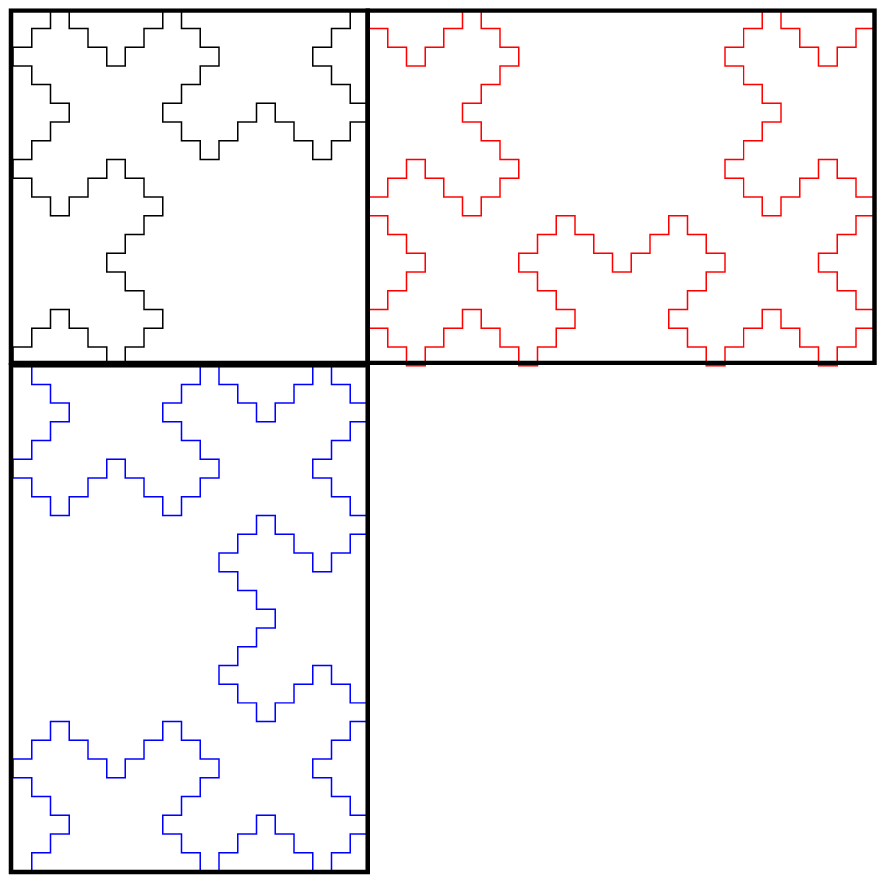}\\
  \caption{$\Sigma_0 q_{12}^{\left[5\right]}$ is divide into parts $\Sigma_0 q_{10}^{\left[5\right]}$, $\Sigma_2 q_{9}^{\left[5\right]}$ and $\Sigma_3 q_{10}^{\left[5\right]}$.}\label{divide}
\end{figure}

By the induction hypothesis $\Sigma_{\alpha_1} q_{n-2}^{\left[i\right]}$ and $\Sigma_{\alpha_2} q_{n-3}^{\left[i\right]}$ are simples, moreover, the three smaller paths are contained in disjoint boxes, then $\Sigma_{\alpha_1} q_{n-2}^{\left[i\right]}$ is simple.

\item If $i$ is even. From Proposition  \ref{pal}, we have $q_{3n}^{\left[i\right]}=p\overline{\sigma_{n-1}}$ for some  antipalindrome $p=w_1\cdots w_n $ and $\overline{\sigma_{n-1}}\in \left\{\texttt{1}, \texttt{3}\right\}$. If $\overline{\sigma_{n-1}}=\texttt{3}$, we can consider the reversal of the path, so suppose that  $\overline{\sigma_{n-1}}=\texttt{1}$. Hence  $\Sigma_\alpha^{\circ}(q_{3n}^{\left[i\right]})^4=\Sigma_{\alpha}(p\texttt{1} \cdot p\texttt{1} \cdot p\texttt{1} \cdot p)$, as $$\Sigma_{\alpha}p\texttt{1}=\alpha\cdot(\alpha + w_1)\cdot(\alpha + w_1 + w_2)\cdots(\alpha + w_1 + w_2 + \cdots w_n+1)$$ and $|p|_1=|p|_3$,  because $p$ is an antipalindrome, then
\begin{align*}
\alpha + w_1 + w_2 + \cdots w_n+1=\alpha + |p|_1+ 3|p|_3+1=\alpha+4|p|_1+1\cong \alpha +1 \mod 4.
\end{align*}
Therefore
\begin{align*}
\Sigma_\alpha^{\circ}(q_{3n}^{\left[i\right]})^4=\Sigma_{\alpha}(p\texttt{1} \cdot p\texttt{1} \cdot p\texttt{1} \cdot p)=\Sigma_\alpha p  \cdot \Sigma_{\alpha+1} p \cdot \Sigma_{\alpha+2} p  \cdot \Sigma_{\alpha+3}p.
\end{align*}
But, the initial segments in the paths  $\Sigma_\alpha p$ and $\Sigma_{\alpha+1} p$ are orthogonal because $\alpha$ and $\alpha+1$ represent orthogonal vectors. Hence  $\Sigma_\alpha p  \cdot \Sigma_{\alpha+1} p \cdot \Sigma_{\alpha+2} p  \cdot \Sigma_{\alpha+3}p$ is a closed polygonal path, illustrated in Fig. \ref{closed} with an angle of $\pi/2$  counterclockwise.
\begin{figure}[!h]
\centering
 \includegraphics[scale=0.5]{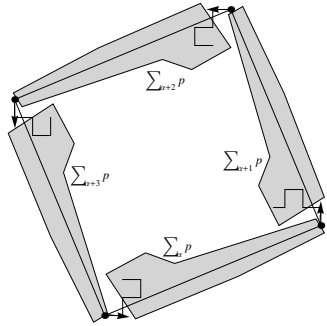}\\
  \caption{Case $ii$ with an angle of $\pi/2$.}\label{closed}
\end{figure}
\item If $i$ is odd, the proof is similar.
\end{enumerate}
\end{proof}

An \emph{$i$-generalized Fibonacci snowflake of order $n$} is a polyomino having $\Sigma_\alpha^{\circ}(q_{3n}^{\left[i\right]})^4$ or $\Sigma_\alpha^{\circ}(q_{3n+2}^{\left[i\right]})^4$  as a boundary word, we denote this as  $\prod_n^{\left[i\right]}$ . In Table \ref{graf4} we show first $i$-generalized Fibonacci snowflakes.
\begin{table}[h]
\scriptsize
\centering
\begin{tabular}{|c|c|c|c|} \hline
 $\prod_1^{\left[2\right]}$  &  $\prod_2^{\left[2\right]}$ & $\prod_3^{\left[2\right]}$ & $\prod_4^{\left[2\right]}$  \\
\includegraphics[scale=0.4]{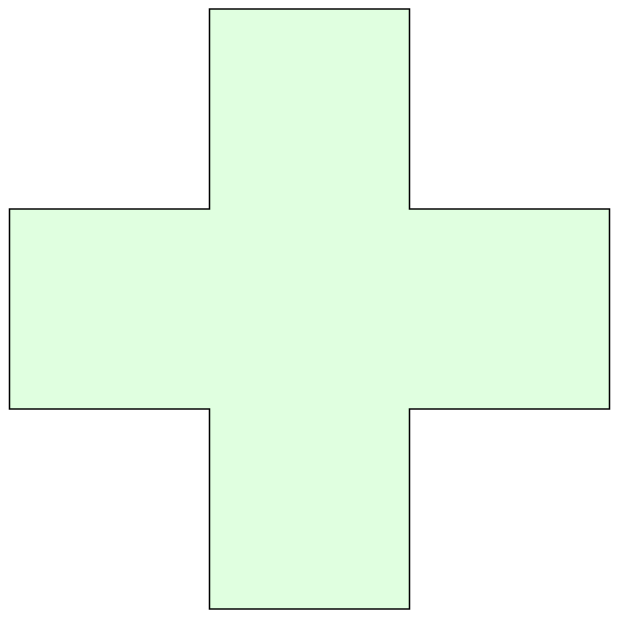} &
\includegraphics[scale=0.4]{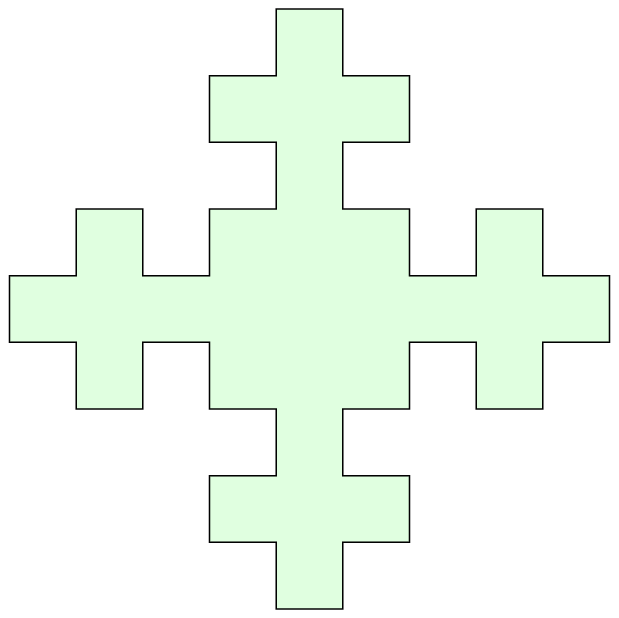} &
\includegraphics[scale=0.4]{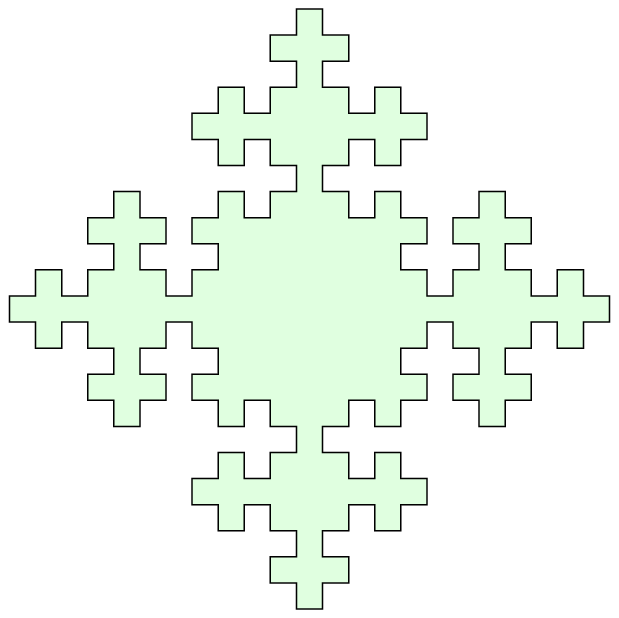} &
\includegraphics[scale=0.4]{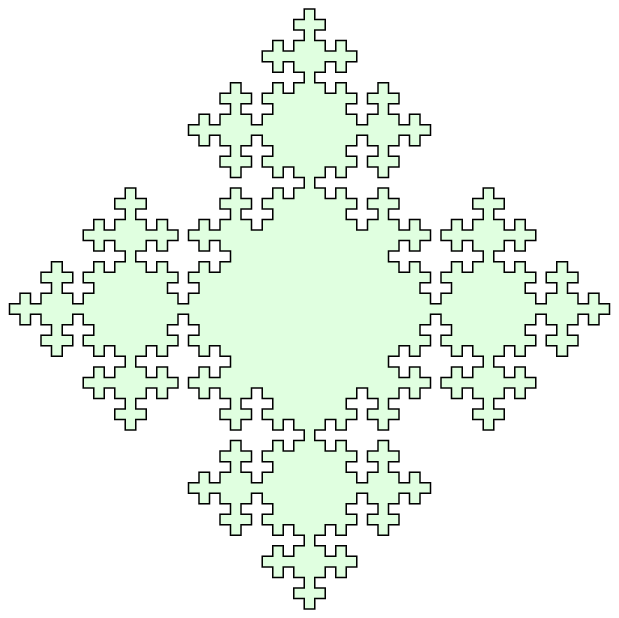} \\ \hline
 $\prod_1^{\left[3\right]}$  &  $\prod_2^{\left[3\right]}$ & $\prod_3^{\left[3\right]}$ & $\prod_4^{\left[3\right]}$  \\
\includegraphics[scale=0.4]{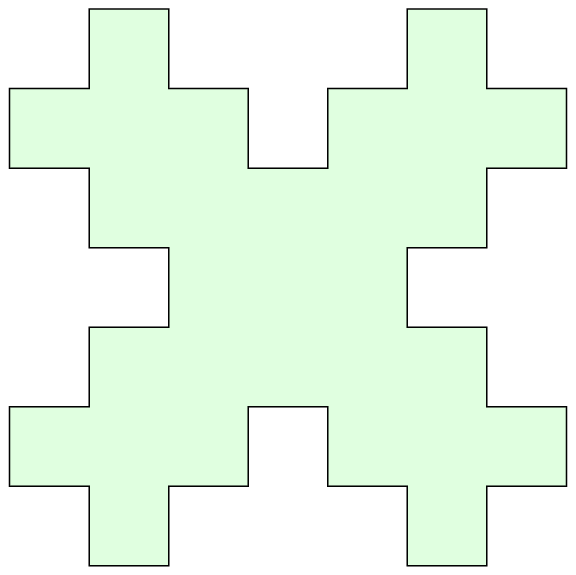} &
\includegraphics[scale=0.4]{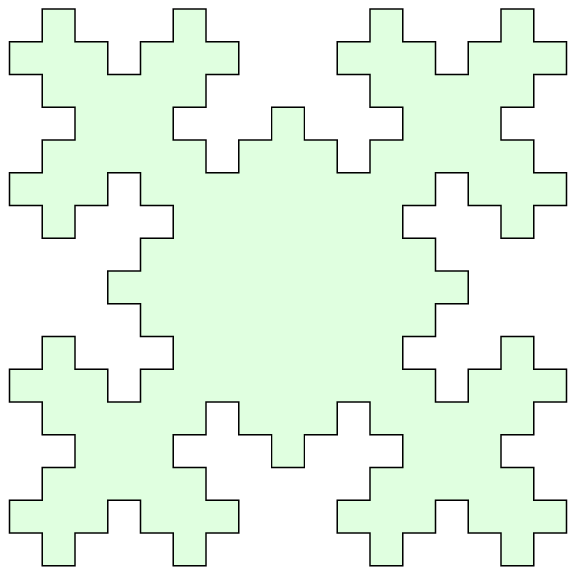} &
\includegraphics[scale=0.4]{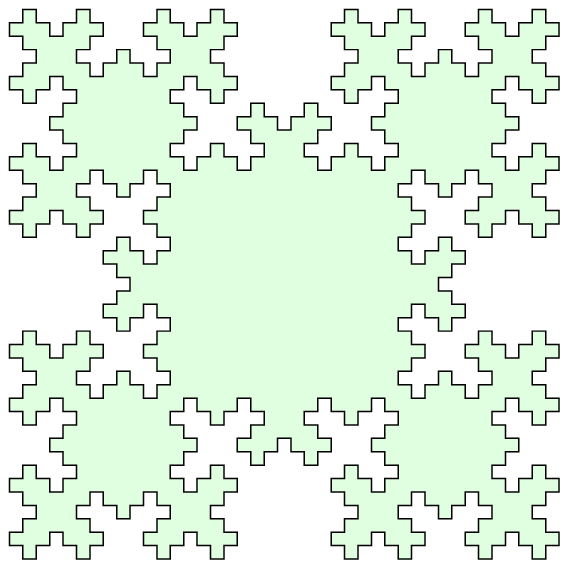} &
\includegraphics[scale=0.4]{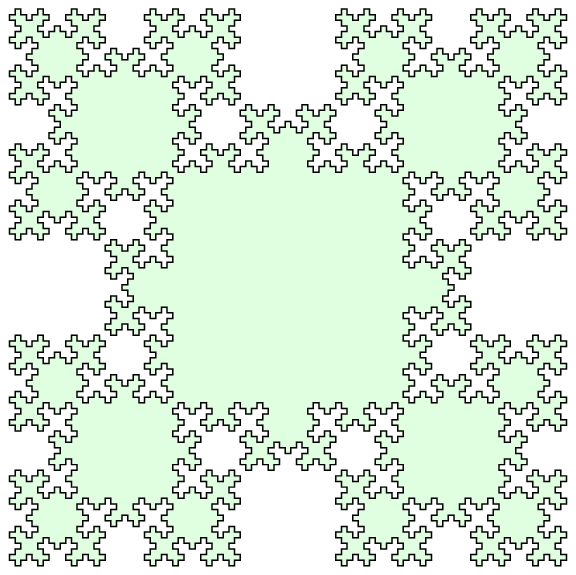} \\ \hline
 $\prod_1^{\left[4\right]}$  &  $\prod_2^{\left[4\right]}$ & $\prod_3^{\left[4\right]}$ & $\prod_4^{\left[4\right]}$  \\
\includegraphics[scale=0.4]{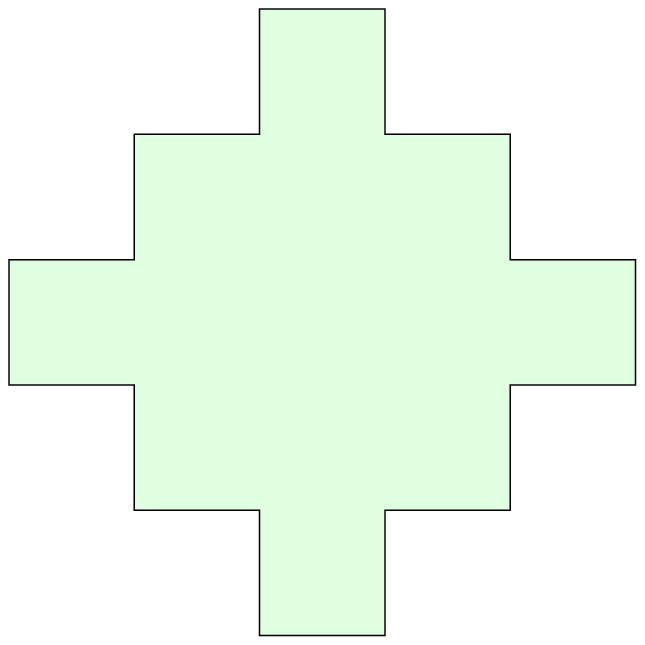} &
\includegraphics[scale=0.4]{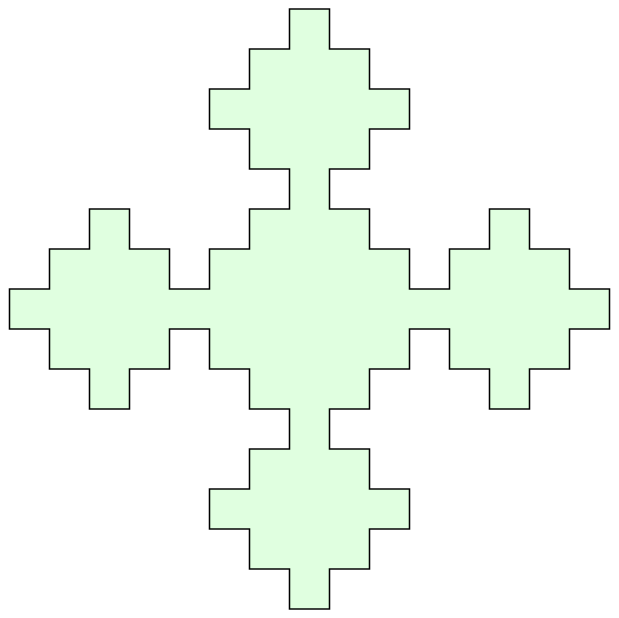} &
\includegraphics[scale=0.4]{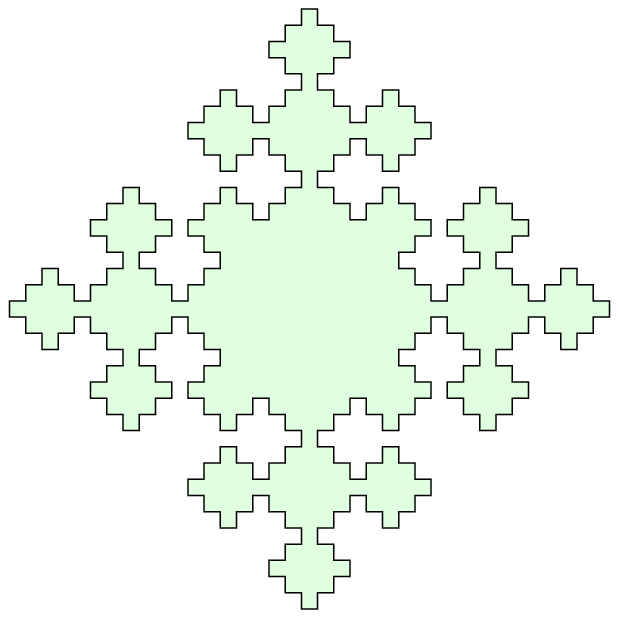} &
\includegraphics[scale=0.4]{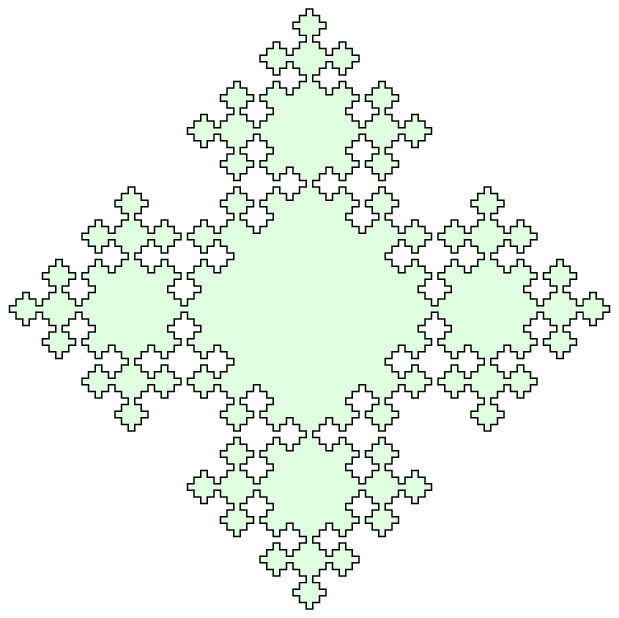} \\ \hline
 $\prod_1^{\left[5\right]}$  &  $\prod_2^{\left[5\right]}$ & $\prod_3^{\left[5\right]}$ & $\prod_4^{\left[5\right]}$  \\
\includegraphics[scale=0.4]{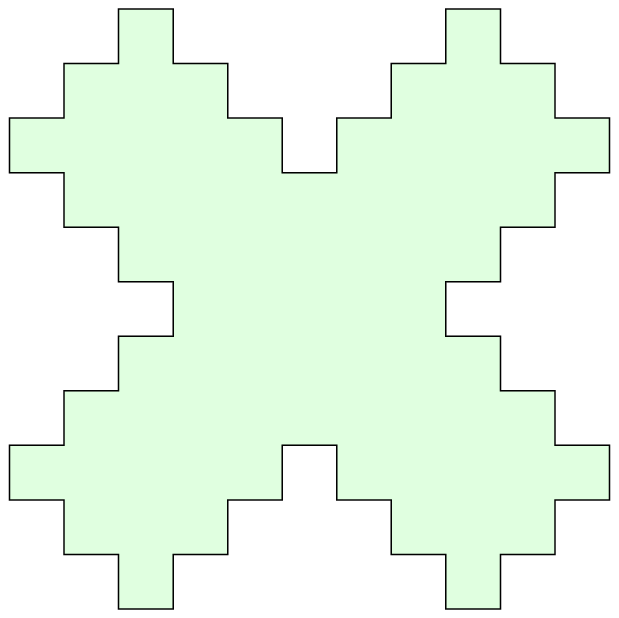} &
\includegraphics[scale=0.4]{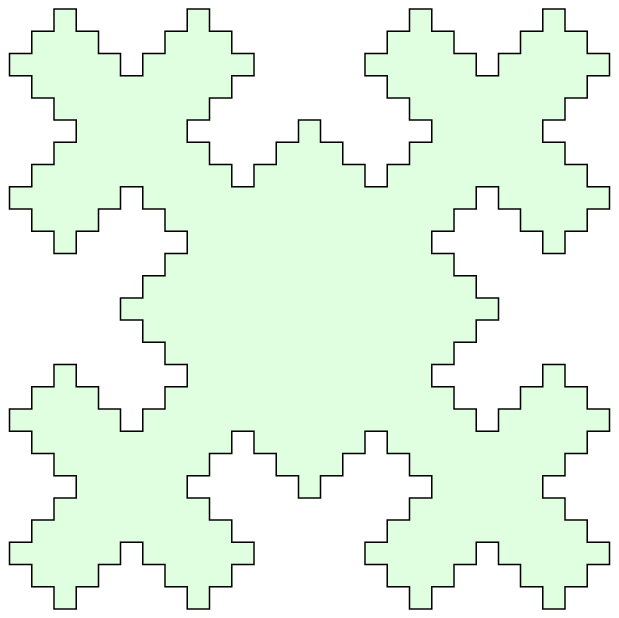} &
\includegraphics[scale=0.4]{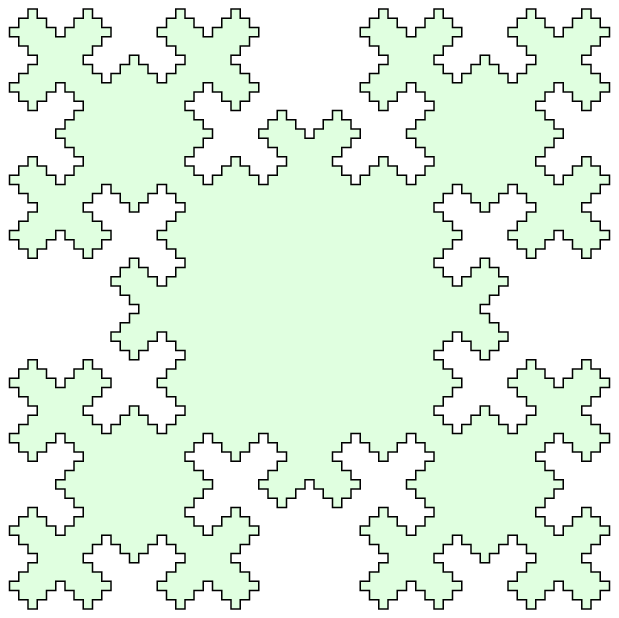} &
\includegraphics[scale=0.4]{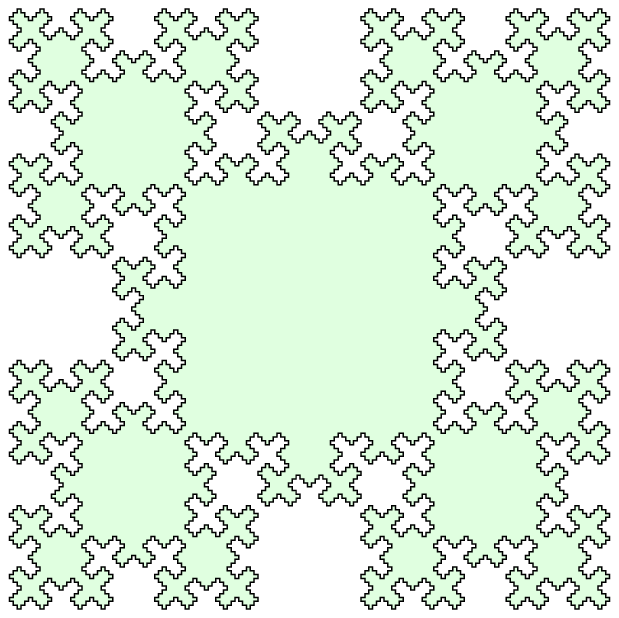} \\ \hline
 $\prod_1^{\left[6\right]}$  &  $\prod_2^{\left[6\right]}$ & $\prod_3^{\left[6\right]}$ & $\prod_4^{\left[6\right]}$  \\
\includegraphics[scale=0.4]{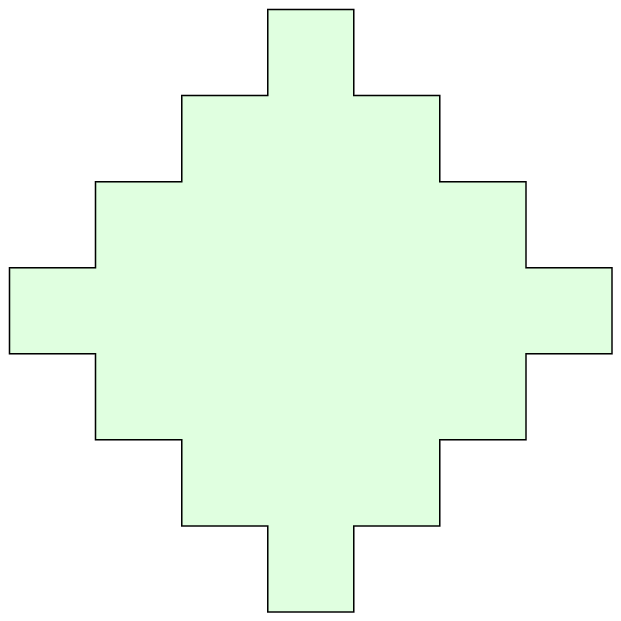} &
\includegraphics[scale=0.4]{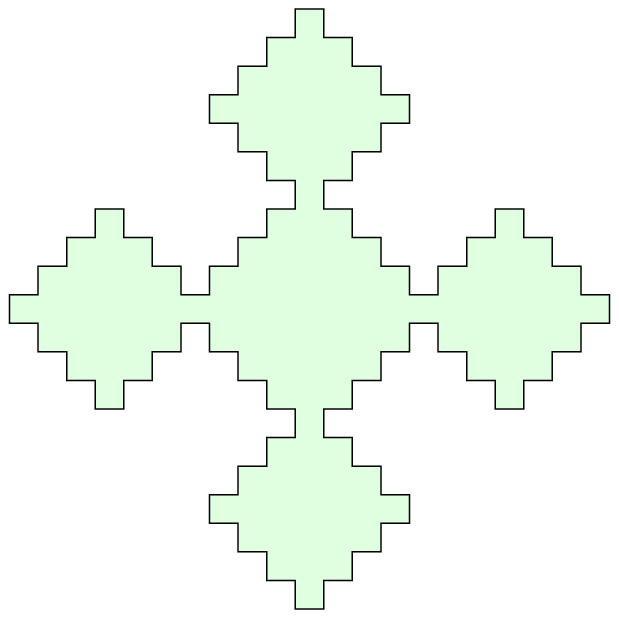} &
\includegraphics[scale=0.4]{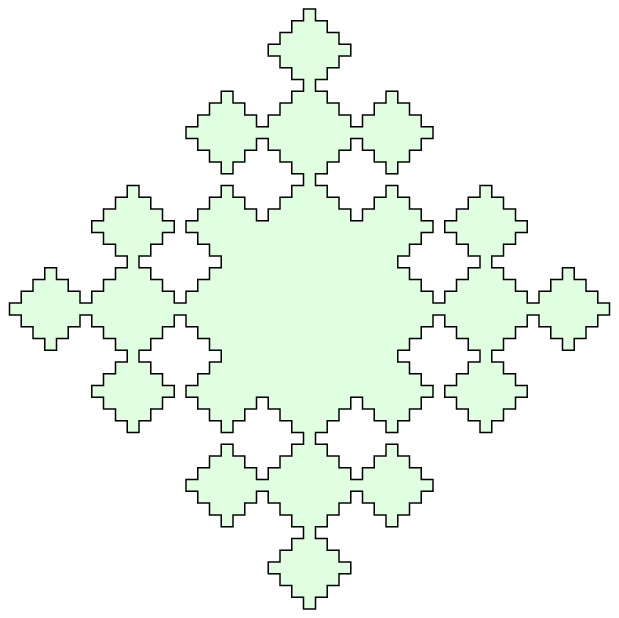} &
\includegraphics[scale=0.4]{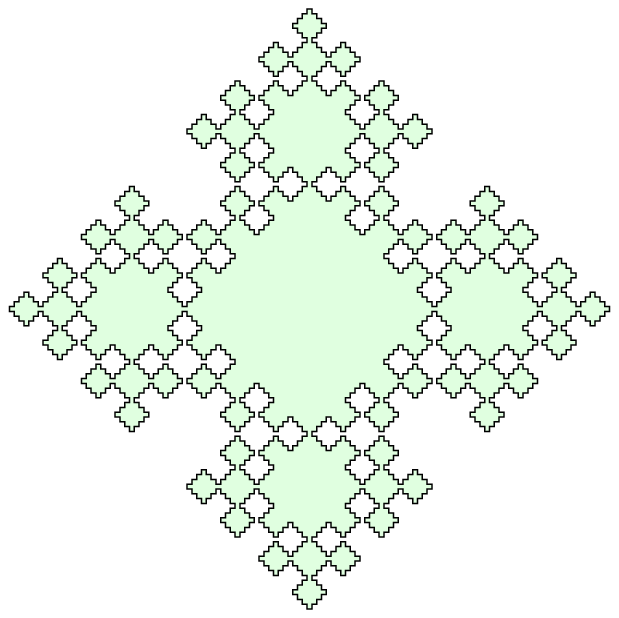} \\ \hline
\end{tabular}
\caption{The $i$-Generalized Fibonacci Snowflakes
  $\prod_n^{\left[i\right]}$ for $i=2, 3, 4, 5, 6$ and $n=1, 2, 3, 4$.}
\label{graf4}
\end{table}

\begin{theorem}\label{poldobles}
 The $i$-generalized Fibonacci snowflake  of order $n\geq 1$ is a double square, for all positive integers  $i$.
\end{theorem}
\begin{proof}
Suppose that  $i$ even. We show in Proposition    \ref{poliminofibo}-$ii$ that
\begin{align*}
\Sigma_\alpha^{\circ}(q_{3n}^{\left[i\right]})^4=\Sigma_{\alpha}(p\texttt{1} \cdot p\texttt{1} \cdot p\texttt{1} \cdot p)=\Sigma_\alpha p  \cdot \Sigma_{\alpha+1} p \cdot \Sigma_{\alpha+2} p  \cdot \Sigma_{\alpha+3}p.
\end{align*}
Moreover $w_j=-w_{n-(j-1)}$, for all $j$ with $1\leqslant j \leqslant n$, because $p$ is an antipalindrome. Then
\begin{align*}\Sigma_{\alpha+2}p&=(\alpha +2)(\alpha +2 + w_1) \cdots (\alpha +2 + w_1 + w_2 + \cdots w_n)\\
&= (\alpha +2 + w_1 + w_2 + \cdots w_n)(\alpha +2 + w_1 + w_2 + \cdots w_{n-1}) \cdots (\alpha + 2)\\
&= \widehat{\Sigma_{\alpha} p}.
\end{align*}
Hence
\begin{align*}
\Sigma_\alpha^{\circ}(q_{3n}^{\left[i\right]})^4=\Sigma_\alpha p  \cdot \Sigma_{\alpha+1} p \cdot \Sigma_{\alpha+2} p  \cdot \Sigma_{\alpha+3}p=\Sigma_\alpha p  \cdot \Sigma_{\alpha+1} p \cdot \widehat{\Sigma_{\alpha} p}  \cdot \widehat{\Sigma_{\alpha+1}p}.
\end{align*}

By the other hand, the word  $q_{3n}^{' \left[i\right]}=\overline{q_{3n-2}^{\left[i\right]} } q_{3n-1}^{\left[i\right]}$ corresponds to another boundary word of the same title. In fact, by  Proposition \ref{pal}, we have $q_{3n-1}^{\left[i\right]}=m\texttt{1}$ and $q_{3n-2}^{\left[i\right]}=r\texttt{3}$, for some palindromes $m$ and $r$. Hence $p\texttt{1}=q_{3n}^{\left[i\right]}=q_{3n-1}^{\left[i\right]}\overline{q_{3n-2}^{\left[i\right]}}=m\texttt{1}\overline{r}\texttt{1}$, so that $p=m\texttt{1}\overline{r}$.

 Therefore
$$ q_{3n}^{' \left[i\right]}=\overline{q_{3n-2}^{\left[i\right]} } q_{3n-1}^{\left[i\right]} = \overline{r}\texttt{1}m\texttt{1}=p^R\texttt{1}=\overline{p}\texttt{1}$$
and $\Sigma_\alpha^{\circ}(q_{3n}^{'  \left[i\right]})^4=\Sigma_\alpha( \overline{p} \texttt{1} \cdot\overline{p} \texttt{1} \cdot p^R \texttt{1}  \cdot p^R)=\Sigma_\alpha \overline{p}  \cdot \Sigma_{\alpha+1} \overline{p} \cdot \widehat{\Sigma_{\alpha} \overline{p}}  \cdot \widehat{\Sigma_{\alpha+1}\overline{p}}.$  \qedhere
\end{proof}

\textbf{Remark}.  Note that if $A\cdot B \cdot \widehat{A}\cdot \widehat{B}$ is a BN-factorization of an $i$-generalized Fibonacci snowflake, then $A$ and $B$ are palindromes, because $p$ is an antipalindrome then  $ \Sigma_\alpha p$ and $ \Sigma_\alpha \overline{p}$ are palindromes.

\begin{example}
In Table \ref{tesse}, we show  tessellations of  $\prod_2^{\left[3\right]}$ and  $\prod_3^{\left[6\right]}$.
\end{example}

\begin{table}[h]
\centering
\begin{tabular}{cc}
$\prod_2^{\left[3\right]}$  & $\prod_3^{\left[6\right]}$  \\
 \includegraphics[scale=0.6]{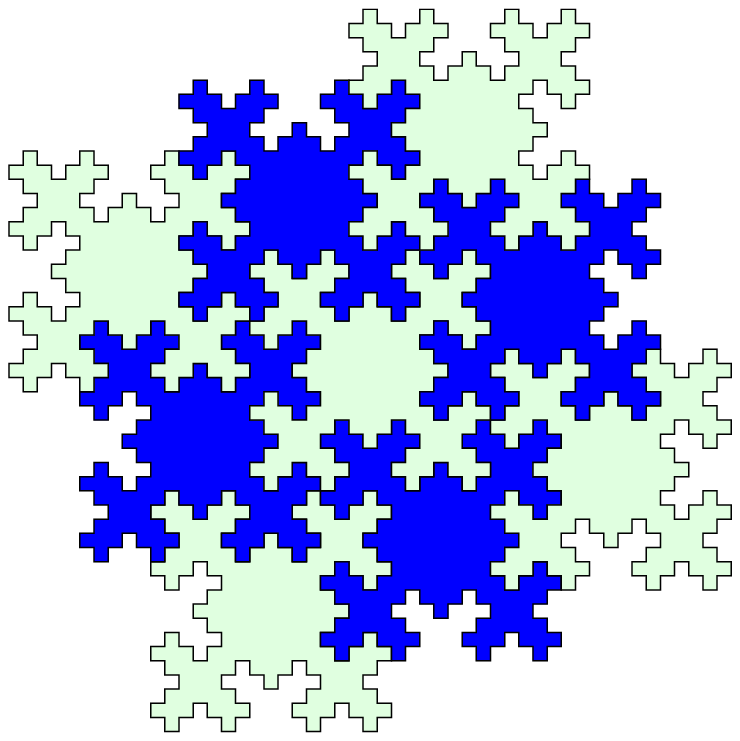} &
    \includegraphics[scale=0.6]{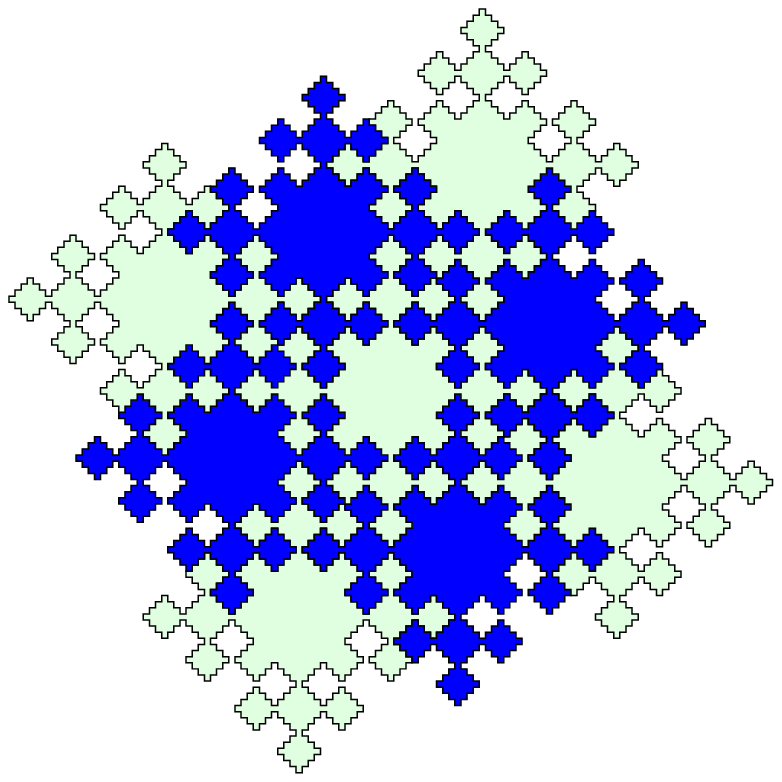}  \\
\end{tabular}
\caption{Tessellations of  $\prod_2^{\left[3\right]}$ and  $\prod_3^{\left[6\right]}$.}
\label{tesse}
\end{table}

\subsection{Some Geometric Properties}

\begin{definition}
The number $P^{\left[i\right]}(n)$ is defined recursively by $P^{\left[i\right]}(0)=-i$, $P^{\left[i\right]}(1)=i+1$ and $P^{\left[i\right]}(n)=2P^{\left[i\right]}(n-1) + P^{\left[i\right]}(n-2)$ for all $n\geq 2$ and $i\geq 0$.
\end{definition}
For $i=0$ we have Pell numbers. In Table \ref{numpell} are the first numbers  $P^{\left[i\right]}(n)$.
\begin{table}[H]
\centering
\begin{tabular}{|c|ll|}  \hline
  $i$ & \multicolumn{2}{c|}{  $P^{\left[i\right]}(n)$} \\ \hline
  0 & $\left\{0, 1, 2, 5, 12, 29, 70, 169, 408, 985, 2378,...\right\}$, &  (A000129). \\ \hline
  1 & $\left\{-1, 2, 3, 8, 19, 46, 111, 268, 647, 1562, 3771,...\right\}$, & (A078343). \\ \hline
  2 & $\left\{-2, 3, 4, 11, 26, 63, 152, 367, 886, 2139, 5164,...\right\}$. & \\ \hline
  3 & $\left\{-3, 4, 5, 14, 33, 80, 193, 466, 1125, 2716, 6557,...\right\}$. & \\ \hline
  4 & $\left\{-4, 5, 6, 17, 40, 97, 234, 565, 1364, 3293, 7950,...\right\}$. &\\ \hline
  5 &  $\left\{-5, 6, 7, 20, 47, 114, 275, 664, 1603, 3870, 9343,...\right\}$. & \\  \hline
\end{tabular}
\caption{First numbers $P^{\left[i\right]}(n)$.}\label{numpell}
\end{table}

\begin{proposition}
A formula for the  $P^{\left[i\right]}(n)$  numbers is
\begin{align*}
P^{\left[i\right]}(n)=\frac{1}{4}\left(\left(1+\sqrt2\right)^{n}(\sqrt{2}-(2-2\sqrt{2})i) - \left(1-\sqrt2\right)^{n}(\sqrt{2}+(2+2\sqrt{2})i) \right).
\end{align*}
\end{proposition}
\begin{proof}
By induction on $n$.
\end{proof}
Let $\alpha\in \mathcal{A}$, we denote by $\stackrel{\rightarrow}{\Sigma_\alpha}q$ the coordinates of the vector whose initial point is the origin  and the terminal point  is the last point of the path  $\Sigma^{\circ}_{\alpha}(q)$. In  the next proposition, we show that the coordinates of the vector $\stackrel{\rightarrow}{\Sigma}_{0}(q_n^{\left[i\right]})$ are expressed in terms of the numbers $P^{\left[i\right]}(n)$. A similar thing happens when $\alpha=1, 2, 3$.

\begin{proposition}\label{coord}
For all $n\in \mathbb{N}$, we have that if $i$ is even then
\begin{align*}
\stackrel{\rightarrow}{\Sigma_0}q_{3n+1}^{\left[i\right]}&=\left(P^{\left[k\right]}(n+1) + P^{\left[k\right]}(n), 0\right), \\
\stackrel{\rightarrow}{\Sigma_0}q_{3n+2}^{\left[i\right]}&=\left(P^{\left[k\right]}(n+1), (-1)^nP^{\left[k\right]}(n+1)\right), \\
\stackrel{\rightarrow}{\Sigma_0}q_{3n+3}^{\left[i\right]}&=\left(P^{\left[k\right]}(n+2), (-1)^{n}P^{\left[k\right]}(n+1)\right),
\end{align*}
where $k=\frac{i-2}{2}$. If $i$ is odd then
\begin{align*}
\stackrel{\rightarrow}{\Sigma_0}q_{3n+1}^{\left[i\right]}&=\begin{cases}
(P^{\left[k\right]}(n+1)+ P^{\left[k\right]}(n), 0), & \text{if $n$ is even},\\
(0, P^{\left[k\right]}(n+1)+ P^{\left[k\right]}(n)), & \text{if $n$ is odd},\\
\end{cases}\\
\stackrel{\rightarrow}{\Sigma_0}q_{3n+2}^{\left[i\right]}&=
\begin{cases}
(P^{\left[k\right]}(n+2), P^{\left[k\right]}(n+1)), & \text{if $n$ is even},\\
(P^{\left[k\right]}(n+1), P^{\left[k\right]}(n+2)), & \text{if $n$ is odd},\\
\end{cases}\\
\stackrel{\rightarrow}{\Sigma_0}q_{3n+3}^{\left[i\right]}&=\left(P^{\left[k\right]}(n+2) , P^{\left[k\right]}(n+2)\right),
\end{align*}
where $k=\frac{i-3}{2}$.
\end{proposition}
\begin{proof}
By induction on $n$. If $i$ is even. For $n=0$ it is clear. Assume for all $j$ such that  $0\leq j \leq 3n +5$; we prove it for $3n+6$. Then passing to vectors we have
\begin{align*}
\stackrel{\rightarrow}{\Sigma_0}q_{3n+6}^{\left[i\right]}&=\stackrel{\rightarrow}{\Sigma_0}q_{3n+5}^{\left[i\right]} + \stackrel{\rightarrow}{\Sigma_0}\overline{q_{3n+4}^{\left[i\right]}}\\
&=\left(P^{\left[k\right]}(n+2), (-1)^{n+1}P^{\left[k\right]}(n+2)\right) + \overline{\left(P^{\left[k\right]}(n+2) + P^{\left[k\right]}(n+1), 0\right)}\\
&=\left(P^{\left[k\right]}(n+2), (-1)^{n+1}P^{\left[k\right]}(n+2)\right) + \left(P^{\left[k\right]}(n+2) + P^{\left[k\right]}(n+1), 0\right)\\
&=\left(2P^{\left[k\right]}(n+2)+ P^{\left[k\right]}(n+1), (-1)^{n+1}P^{\left[k\right]}(n+2)\right)\\
&=\left(P^{\left[k\right]}(n+3), (-1)^{n+1}P^{\left[k\right]}(n+2)\right)
\end{align*}
where $\stackrel{\rightarrow}{\Sigma_0}\overline{q_{n}^{\left[i\right]}}=(\overline{A}, \overline{B})$ is the  coordinate the last point of the path  $\Sigma^{\circ}_{\alpha}(\overline{q}_n)$. In this case \linebreak  $\stackrel{\rightarrow}{\Sigma_0}\overline{q_{3n+4}^{\left[i\right]}}=\stackrel{\rightarrow}{\Sigma_0}q_{3n+4}^{\left[i\right]}$, because $\overline{a}$ leaves the horizontal direction unchanged. The other cases are similar.
\end{proof}

\begin{example}
Table \ref{coordino} are the endpoints coordinates of the paths   $\Sigma^{\circ}_{0}(q_n^{\left[4\right]})$ and Fig. \ref{coordino1} shows the coordinates.

\begin{table}[h]
\begin{center}
\begin{tabular}{|c|c|c|c|c|c|} \hline
$n$ & 0 & 1 & 2 & 3 & 4   \\ \hline
$\Sigma^{\circ}_{0}(q_{3n+1}^{\left[4\right]})$ & (1, 0) & (5,0) & (11,0) & (27, 0) & (65, 0)  \\ \hline
$\Sigma^{\circ}_{0}(q_{3n+2}^{\left[4\right]})$ & (2, 2) & (3,-3) & (8,8) & (19, -19) & (46, 6)  \\ \hline
$\Sigma^{\circ}_{0}(q_{3n+3}^{\left[4\right]})$ & (3, 2) & (8,-3) & (19,8) & (46, -19) & (111, 46)  \\ \hline
\end{tabular}
\end{center}
\caption{Coordinates of the path $\Sigma^{\circ}_{0}(q_n^{\left[4\right]})$ . }
\label{coordino}
\end{table}

\begin{figure}[!h]
\centering
 \includegraphics[scale=0.55]{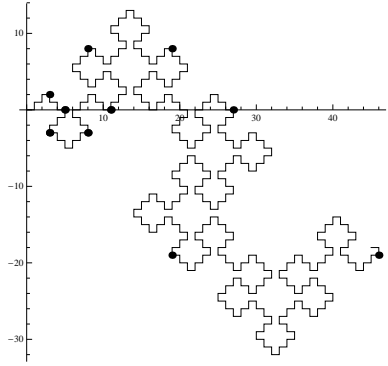}\\
  \caption{Graph with the coordinates of the path $\Sigma^{\circ}_{0}(q_n^{\left[4\right]})$.}\label{coordino1}
\end{figure}

\end{example}
The following proposition is clear because $|q_n^{\left[i\right]}|=F_{n-1}^{\left[i\right]}$.

\begin{proposition}\label{perimetro}
The perimeter $L(n, i)$  of the $i$-generalized Fibonacci snowflake of order $n$ is
\begin{align*} L(n, i)=
\begin{cases}
4F_{3n-1}^{\left[i\right]}, & \text{if $i$ is even}  \\
4F_{3n+1}^{\left[i\right]}, & \text{if $i$ is odd.}
\end{cases}
\end{align*}
\end{proposition}

\begin{proposition}\label{area}
 The area  $A(n, i)$ of  the $i$-generalized Fibonacci snowflake of order $n$ is:
 \begin{enumerate}[i.]
  \item If $i$ is even, then $A(n, i)=\left(P^{\left[k\right]}(n+1)\right)^2 + \left(P^{\left[k\right]}(n)\right)^2$,  where $k=\frac{i-2}{2}$.

 \item If $i$ is odd then, $A(n, i)=\left(P^{\left[k\right]}(n+2)\right)^2 + \left(P^{\left[k\right]}(n+1)\right)^2$, where $k=\frac{i-3}{2}$.

 \item Moreover $A(n, i)$ satisfies the recurrence  formula
 \begin{align}
 A(n, i)=6A(n-1, i) - A(n-2, i)  \label{eqarea}
 \end{align}
 for all $n\geq 3$, (initial values can be calculated with the above items).
\end{enumerate}
\end{proposition}
\begin{proof}
Suppose that $i$ is even.  If a word $w\in \mathcal{A}^*$  is an antipalindrome then its corresponding polygonal line is symmetric with respect to midpoint of the vector  $\stackrel{\rightarrow}{\Sigma_\alpha}w$,  see Lemma 2.6 in \cite{FIB}. Moreover, from  Proposition \ref{poliminofibo}-ii, we have that the parallelogram determined by the word $\Sigma_\alpha^{\circ}(q_{3n}^{\left[i\right]})^4$ is a square, (in Fig.  \ref{areas}, we show some examples for $i=2, 3, 4$ and $n=2$), and by Proposition \ref{coord} the  area  $A(n, i)$ is equal to the area of square determined by  $\Sigma^{\circ}_{0}(q_{3n}^{\left[i\right]})=\left(P^{\left[k\right]}(n+1), \pm P^{\left[k\right]}(n)\right)$.  Hence $A(n, i)=\left(P^{\left[k\right]}(n+1)\right)^2 + \left(P^{\left[k\right]}(n)\right)^2$,  where $k=\frac{i-2}{2}$. If $i$ is odd, the proof is similar.\\

The Eq. \ref{eqarea} is obtained from $i$ and $ii$, and by  definition of $P^{\left[i\right]}(n)$.

\begin{figure}[!h]
\centering
 \includegraphics[scale=0.3]{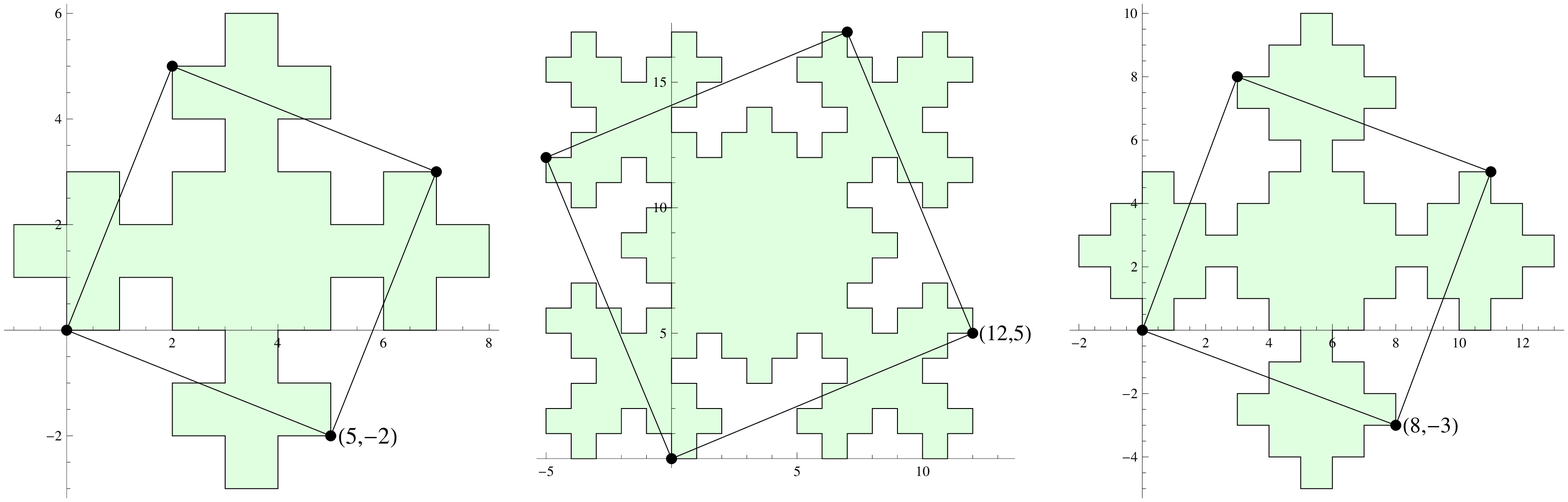}\\
  \caption{Examples, Areas of $i$-generalized Fibonacci snowflakes.}\label{areas}
\end{figure}
\end{proof}

Let $S^{\left[i\right]}(n)$ be the  smallest square having sides parallel to the axes and containing to $\prod^{\left[i\right]}$. In Fig. \ref{casco}, we show the cases for  $i=4$ and $n=2, 3$. If $i$ is even, from Proposition \ref{coord} we have that  $(A, B)=(P^{\left[i\right]}(n), (-1)^nP(n+1)^{\left[i\right]})$. Therefore
\begin{align*}
S^{\left[i\right]}(n)=\left(\frac{A+3B}{2}-\frac{A-B}{2}-1\right)^2=(2B-1)^2=(2P^{\left[i\right]}(n+1)-1)^2
\end{align*}
When $i$ is odd it is similar.

\begin{figure}[!h]
\centering
 \includegraphics[scale=0.27]{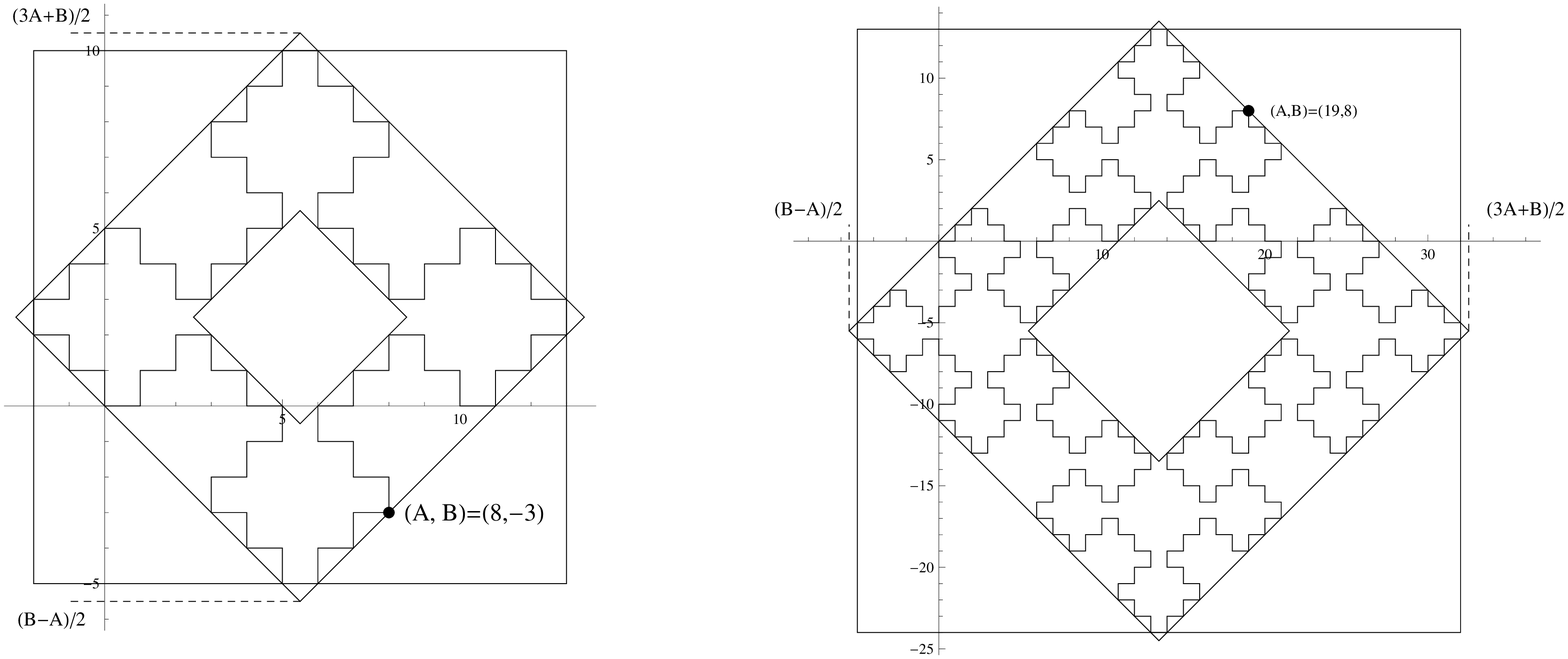}\\
  \caption{ $S^{\left[i\right]}(n)$ for $i=4$ and $n=2, 3$.}\label{casco}
\end{figure}

Next theorem generalizes  theorem 1 of \cite{BLO4}.
\begin{theorem}
The  fractal dimension of $\prod^{\left[i\right]}=\lim_{n\rightarrow \infty}\prod_n^{\left[i\right]}$ is
\begin{align*}
\frac{3\ln \phi}{\ln(1+\sqrt{2}) }.
\end{align*}
\end{theorem}
\begin{proof}
Suppose that  $i$ is even, then the polyomino  $\prod_n^{\left[i\right]}$ is composed of $4|q_{3n}^{\left[i\right]}|$ unit segments and this value blows up when $n\rightarrow \infty$. However, the normalized polyomino $\frac{1}{2P^{\left[i\right]}(n+1)-1}\prod_n^{\left[i\right]}$ stays  bounded. It has $4|q_{3n}^{\left[i\right]}|$ segments of  length  $\frac{1}{2P^{\left[i\right]}(n+1)-1}$. Hence the total  $d-$dimensional normalized polyomino has length
\begin{align*}
\frac{4|q_{3n}^{\left[i\right]}|}{(2P^{\left[i\right]}(n+1)-1)^d}
\end{align*}
and therefore the self-similarity dimension (see \cite{PEIT} for the definition the self-similarity dimension) of  $\prod^{\left[i\right]}$ is
\begin{align*}
d=\lim_{n\rightarrow \infty}\frac{\ln(4|q_{3n}^{\left[i\right]}|)}{\ln(2P^{\left[i\right]}(n+1)-1)}=\frac{3\ln \phi}{\ln(1+\sqrt{2}) }.
\end{align*}
\end{proof}

\section{Conclusion}
In this paper, we study a generalization of the Fibonacci word and the Fibonacci word fractal founds in \cite{ALE}. Particularly , we defined the curves  $\mathcal{F}^{\left[i\right]}$  from the $i$-Fibonacci words and show their properties remain.
Moreover, the $i$-generalized Fibonacci snowflakes generalize the Fibonacci snowflake studied in \cite{FIB} and we show that they are a subclass of double squares.  Finally, we found that $i$-generalized Fibonacci snowflakes  are related with Fibonacci  and Pell numbers, and some generalizations.

In \cite{EDS} authors have introduced a generalization of the Fibonacci sequence.  For any two nonzero real numbers $a$ and $b$, the \emph{generalized Fibonacci sequence}, say $\left\{F_n^{(a,b)}\right\}_0^{\infty}$, is defined recursively by
\begin{align*}
&F_0^{(a,b)}=0, \ \ \  F_1^{(a,b)}=1, \\
& F_n^{(a,b)}=\begin{cases}
aF_{n-1}^{(a,b)}+F_{n-2}^{(a,b)},  & \ \text{if $n$ is even} \\
bF_{n-1}^{(a,b)}+F_{n-2}^{(a,b)},  & \ \text{if $n$ is odd}
\end{cases} (n\geq2)
\end{align*}
On the other hand, there is a word-combinatorial interpretation of this generalized Fibonacci sequence.  Let $\alpha=\left[0,a,b,a,b,\ldots\right]=\left[0,\overline{a,b}\right]$ then  $\textbf{\emph{w}}(\alpha)=\lim_{n\rightarrow\infty}s_n$ where
\begin{align*}
&s_0=\texttt{1}, \ \ s_1=\texttt{0},  \ \ s_2=\texttt{0}^{a-1}\texttt{1}, \\
&s_n=\begin{cases}s_{n-1}^as_{n-2}, & \text{if $n$ is even} \\  s_{n-1}^bs_{n-2}, & \text{if $n$ is odd} \end{cases}, n\geq 3
\end{align*}
 Let $r_0=0, \ \ r_n=|s_n|, \ n\geq1$ then $\left\{r_n\right\}=\left\{F_n^{(a,b)}\right\}$. It would be interesting to study different curves obtained by applying the odd-even drawing rule to the word $s_n$. Empirical observations  show interesting patterns. For instance with $a=2, b=5$ and $n=9$ we obtain the curve Fig. \ref{newfibo}.

 \begin{figure}[!h]
\centering
 \includegraphics{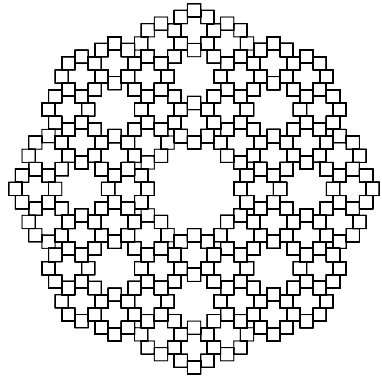}\\
  \caption{ Curve obtained with $a=2, b=5$ and $n=9$.}\label{newfibo}
\end{figure}


\begin{thebibliography}{99} \small

\bibitem{SHA2} J. Allouche, J. Shallit, Automatic Sequences, Cambridge University Press, Cambridge, 2003.

\bibitem{NIV} D. Beauquier, M. Nivat, On translating one polyomino to tile the plane, Discrete Comput. Geom. \textbf{6} (1991), 575--592.

\bibitem{BER} J. Berstel, Fibonacci words-a survey, in: G. Rosenberg, A. Salomaa (Eds.),  The Book of L, Springer, Berlin, (1986),  11--26.

\bibitem{BLO2} A. Blondin-Massé, S. Brlek, A. Garon, S. Labbé, Two infinite families of polyominoes that tile the plane by translation in two distinct ways, Theoret. Comput. Sci. \textbf{412}(2011), 4778--4786.

\bibitem{BLO} A. Blondin-Massé, S. Brlek, S. Labbé, A parallelogram tile fills the plane by translation in at most two distinct ways, Discrete Appl. Math. \textbf{160}(2012), 1011--1018.

\bibitem{BLO4} A. Blondin-Massé, S. Brlek, S. Labbé, M.  Mendès France, Complexity of the Fibonacci snowflake, Fractals \textbf{20} (2012), 157--260.

\bibitem{FIB} A. Blondin-Massé, S. Brlek, S. Labbé, M.  Mendès France, Fibonacci snowflakes, Ann. Sci. Math. Québec \textbf{35}(2)(2010), 141--152.

\bibitem{BLO3} A. Blondin Massé, A. Garon, S. Labbé, Combinatorial properties of double square tiles, Theoret. Comput. Sci. \textbf{502}(2013), 98--117.

\bibitem{BRA}  P. Brass, W. Moser, J. Pach, Research Problems in Discrete Geometry, Springer-Verlag, New York, 2005.

\bibitem{BRL2} S. Brlek. Interactions between Digital Geometry and Combinatorics on Words, in: P. Ambro\u{z}, \u{S}. Holub, Z. Masáková (Eds.),  Proc. WORDS 2011,  8th International Conference Words 2011, Prague, Czech Republic, 12-16 September,  EPTCS, vol 63, 2011, 1--12 .

\bibitem{BRL} S. Brlek, J. Fédou, X.  Provençal, On the Tiling by Translation Problem,  Discrete Appl. Math. \textbf{157} (2009), 464--475.

\bibitem{FIB1} J. Cassaigne, On extremal properties of the Fibonacci word, RAIRO - Theor. Inf. Appl. \textbf{42} (4) (2008), 701--715.

\bibitem{FIB2} W. Chuan, Fibonacci words, Fibonacci Quart., \textbf{30}(1) (1992), 68--76.

\bibitem{FIB3} W. Chuan, Generating Fibonacci words, Fibonacci Quart., \textbf{33}(2) (1995), 104--112.


\bibitem{LUCA} A. de Luca, A division property of the Fibonacci word, Inform. Process. Lett., 54 (1995), 307--312.

\bibitem{DRO} X. Droubay, Palindromes in the Fibonacci word,  Inform. Process. Lett., \textbf{55} (1995), 217--221.


 \bibitem{EDS} M. Edson, O. Yayenie,  A new generalization of Fibonacci sequence and extended Binet's formula. Integers, \textbf{9}(6) (2009), 639--654.


\bibitem{GRU} B. Grünbaum, G.C. Shephard, Tilings and Patterns, W.H. Freeman, New York, 1987.

\bibitem{koshy} T. Koshy, Fibonacci and Lucas Numbers with Applications, Wiley-Interscience,  2001.

\bibitem{LOT2} M. Lothaire, Algebraic Combinatorics on Words, Encyclopedia of Mathematics and its Applications, Cambridge University Press, Cambridge, 2002.

\bibitem{FIB5}  F. Mignosi, G.  Pirillo, Repetitions in the Fibonacci infinite word, RAIRO Inform. Theor. Appl. \textbf{26} (1992), 199--204.

\bibitem{ALE} A. Monnerot, The Fibonacci Word Fractal, preprint  \par
\texttt{ http://hal.archives-ouvertes.fr/hal-00367972/fr/}, (2009).

\bibitem{PEIT} H.O. Heitgen, H. Jürgens, D. Saupe,  Chaos and Fractals: New Frontiers of Science, 2nd ed., Springer-Verlag, New York, 2004.

\bibitem{PIR} G. Pirillo, Fibonacci numbers and words, Discrete Math. \textbf{173} (1997), 197--207.

\bibitem{PUB} P. Prusinkiewicz, A. Lindenmayer, The algorithmic beauty of plants, Springer-Verlag. Nueva
York, 2004. \texttt{http://algorithmicbotany.org/papers/abop/abop.pdf.}

\bibitem{RAM} J. Ramírez, G. Rubiano, Generating fractals curves from homomorphisms between languages  $\left[ \right.$with \textsf{Mathematica}$^\circledR \left. \right]$'' (Spanish), Revista Integración \textbf{30}(2), (2012), 129--150.


\bibitem{OEIS} N. Sloane, The On-Line Encyclopedia of Integer Sequences.

\bibitem{WIJ} H.A.G. Wijshoff, J. van Leeuven, Arbitrary versus periodic storage schemes and tesselations of the plane using one type of polyomino, Inform. Control, \textbf{62} (1984), 1--25.

\end{thebibliography}
\end{document}